\documentclass[reqno,12pt]{amsart}

\date{May 19, 2014}

\usepackage[margin=1in]{geometry}
\usepackage{xspace}
\usepackage[utf8]{inputenc}
\usepackage{graphicx}
\usepackage{amssymb}
\usepackage{mathrsfs}
\usepackage[active]{srcltx}
\usepackage{url}
\usepackage{hyperref}
\usepackage{amsmath,amstext,amsxtra,amsgen,amsbsy,amsopn,amscd,amsthm,amsfonts}
\usepackage{latexsym}
\usepackage{dsfont}
\usepackage{enumitem}

\newtheorem{theorem}{Theorem}

\newtheorem{proposition}{Proposition}
\newtheorem{lemma}[proposition]{Lemma}

\theoremstyle{definition}

\theoremstyle{remark}

\newtheorem{remark}[proposition]{Remark}

\newcommand\R{{\ensuremath {\mathbb R} }}
\newcommand\C{{\ensuremath {\mathbb C} }}
\newcommand\N{{\ensuremath {\mathbb N} }}

\renewcommand\phi{\varphi}

\renewcommand\le{\leqslant}
\renewcommand\ge{\geqslant}
\renewcommand\leq{\leqslant}
\renewcommand\geq{\geqslant}
\renewcommand\epsilon{\varepsilon}
\renewcommand\hat{\widehat}
\renewcommand\tilde{\widetilde}
\renewcommand\bar{\overline}

\newcommand{\gH}{\mathfrak{H}}

\newcommand{\gS}{\mathfrak{S}}

\newcommand\ii{{\ensuremath {\infty}}}

\newcommand{\norm}[1]{ \left| \! \left| #1 \right| \! \right| }

\newcommand{\cE}{\mathcal{E}}

\newcommand{\B}{\mathcal{B}}

\newcommand{\cR}{\mathcal{R}}

\newcommand\1{{\ensuremath {\mathds 1} }}
\newcommand{\Sph}{\mathbb{S}}

\DeclareMathOperator{\re}{Re}

\DeclareMathOperator{\sgn}{sgn}

\DeclareMathOperator{\tr}{Tr}
\DeclareMathOperator{\Det}{Det}

\title[Restriction theorems for orthonormal functions]{Restriction theorems for orthonormal functions, Strichartz inequalities, and uniform Sobolev estimates}

\author[R. L. Frank]{Rupert L. FRANK}
\address{R. L. Frank, Mathematics 253-37, Caltech, Pasadena CA 91125, USA} 
\email{rlfrank@caltech.edu}

\author[J. Sabin]{Julien SABIN}
\address{J. Sabin, Universit\'e de Cergy-Pontoise, Mathematics Department (UMR 8088), F-95000 Cergy-Pontoise, France} 
\email{julien.sabin@u-cergy.fr}

\begin{document}
 
\begin{abstract}
  We generalize the theorems of Stein--Tomas and Strichartz about surface restrictions of Fourier transforms to systems of orthonormal functions with an optimal dependence on the number of functions. We deduce the corresponding Strichartz bounds for solutions to Schr\"odinger equations up to the endpoint, thereby solving an open problem of Frank, Lewin, Lieb and Seiringer. We also prove uniform Sobolev estimates in Schatten spaces, extending the results of Kenig, Ruiz, and Sogge. We finally provide applications of these results to a Limiting Absorption Principle in Schatten spaces, to the well-posedness of the Hartree equation in Schatten spaces, to Lieb--Thirring bounds for eigenvalues of Schr\"odinger operators with complex potentials, and to Schatten properties of the scattering matrix.
\end{abstract}

\maketitle

\section*{Introduction}

A classical topic in harmonic analysis is the so-called \emph{restriction problem}. Given a surface $S$ embedded in $\R^N$, $N\ge2$, one asks for which exponents $1\le p\le 2$, $1\le q\le\ii$ the Fourier transform of a function $f\in L^p(\R^N)$ belongs to $L^q(S)$, where $S$ is endowed with its $(N-1)$-dimensional Lebesgue measure $d\sigma$. More precisely, defining the restriction operator $\cR_S$ as $\cR_Sf=\hat{f}\,|_S$ for all $f$ in the Schwartz class, the problem is to know when $\cR_S$ can be extended as a bounded operator from $L^p(\R^N)$ to $L^q(S)$. The operator dual to $\cR_S$ is called the \emph{extension operator}, which we denote by $\cE_S$, and satisfies the identity
\begin{equation}\label{eq:def-ES}
  \cE_Sf(x)=\frac{1}{(2\pi)^{N/2}}\int_Sf(\xi) e^{i\xi\cdot x}\,d\sigma(\xi),\quad\forall x\in\R^N,
\end{equation}
for all $f\in L^1(S)$. The restriction problem is thus equivalent to knowing when $\cE_S$ is bounded from $L^{q'}(S)$ to $L^{p'}(\R^N)$. We refer to \cite{Tao-04} for a wide review of results concerning this problem and its motivations.

A model case of the restriction problem which is often considered in the literature is the case $q=2$. There are two types of surfaces for which this problem has been completely settled. For smooth compact surfaces with non-zero Gauss curvature, the celebrated Stein--Tomas theorem \cite{Stein-86,Tomas-75} states that the restriction problem has a positive answer if and only if $1\le p\le 2(N+1)/(N+3)$. For quadratic surfaces, Strichartz \cite{Strichartz-77} gave a complete answer depending on the type of the surface (paraboloid-like, cone-like, or sphere-like, see below for a more precise definition). Hence, in these cases we know exactly for which exponents $p$ the inequality 
\begin{equation}\label{eq:ES-bounded}
  \norm{\cE_Sf}_{L^{p'}(\R^N)}\le C\norm{f}_{L^2(S)}
\end{equation}
holds for all $f\in L^2(S)$ with $C>0$ independent of $f$. The question we want to address in this work is a generalization of \eqref{eq:ES-bounded} to systems of orthonormal functions. More precisely, let $(f_j)_{j\in J}$ a (possibly infinite) orthonormal system in $L^2(S)$, and let $(\nu_j)_{j\in J}\subset\C$ be a family of coefficients. We prove inequalities of the form
\begin{equation}\label{eq:ES-orthonormal}
 \norm{\sum_{j\in J}\nu_j|\cE_Sf_j|^2}_{L^{p'/2}(\R^N)}\le C\left(\sum_{j\in J}|\nu_j|^\alpha\right)^{1/\alpha}
\end{equation}
for some $1\le p\le2,\alpha>1$, with $C>0$ independent of $(f_j)$, $(\nu_j)$. To appreciate the difference between \eqref{eq:ES-bounded} and \eqref{eq:ES-orthonormal}, notice that combining \eqref{eq:ES-bounded} with the triangle inequality in $L^{p'/2}(\R^N)$ leads to the estimate
$$\norm{\sum_{j\in J}\nu_j|\cE_Sf_j|^2}_{L^{p'/2}(\R^N)}\le \sum_{j\in J}|\nu_j|\norm{\cE_Sf_j}_{L^{p'}(\R^N)}^2\le C\sum_{j\in J}|\nu_j|,$$
which is weaker than \eqref{eq:ES-orthonormal} (since we will prove \eqref{eq:ES-orthonormal} with $\alpha>1$). This difference is particularly important when the number of non-zero $\nu_j$ is infinite, which happens in some applications that we will discuss below. In the case of smooth compact surfaces with non-zero Gauss curvature we will prove \eqref{eq:ES-orthonormal} with the \emph{optimal} (that is, largest possible) exponent $\alpha$.

Generalizing functional inequalities involving a single function to systems of orthonormal functions is not a new topic. It is strongly motivated by the study of many-body systems in quantum mechanics, where a simple description of $M$ independent fermionic particles is given by $M$ orthonormal functions in some $L^2$-space. It is then important to obtain functional inequalities on these systems whose behaviour is optimal in the number $M$ of such functions. Typically, this behaviour is better than the one given by the triangle inequality. The first example of such a generalization is the Lieb--Thirring inequality \cite{LieThi-75}, which states that for any $f_1,\ldots,f_M$ orthonormal in $L^2(\R^d)$ and for any non-negative coefficients $\nu_1,\ldots,\nu_M$, we have
\begin{equation}\label{eq:LT-M-fct}
  \norm{\sum_{j=1}^M\nu_j|f_j|^2}_{L^{1+\frac2d}(\R^d)}\le C\left(\sup_j \nu_j\right)^{\frac{2}{d+2}}\left(\sum_{j=1}^M\nu_j\norm{\nabla f_j}_{L^2(\R^d)}^2\right)^{\frac{d}{d+2}}.
\end{equation}
Its counterpart for a single function is the Gagliardo--Nirenberg--Sobolev inequality,
\begin{equation}\label{eq:LT-1-fct}
  \norm{f}_{L^{2+\frac4d}(\R^d)}\le C\norm{f}_{L^2(\R^d)}^{\frac{2}{d+2}}\norm{\nabla f}_{L^2(\R^d)}^{\frac{d}{d+2}},
\end{equation}
which together with the triangle inequality implies
\begin{equation}\label{eq:LT-triangle}
  \norm{\sum_{j=1}^M\nu_j|f_j|^2}_{L^{1+\frac2d}(\R^d)}\le C\left(\sum_{j=1}^M\nu_j\right)^{\frac{2}{d+2}}\left(\sum_{j=1}^M\nu_j\norm{\nabla f_j}_{L^2(\R^d)}^2\right)^{\frac{d}{d+2}},
\end{equation}
which is weaker than \eqref{eq:LT-M-fct}. Lieb--Thirring inequalities are a decisive tool for proving stability of matter \cite{LieThi-75,Lieb-90-book}. In particular, we emphasize that \eqref{eq:LT-triangle} is \emph{not} enough to prove stability of matter. The homogeneous Sobolev inequality for $0<s<d/2$,
\begin{equation}\label{eq:hom-Sob-1-fct}
 \norm{(-\Delta)^{-s/2}f}_{L^{\frac{2d}{d-2s}}(\R^d)}\le C\norm{f}_{L^2(\R^d)},
\end{equation}
also has a generalization to several functions which was proved by Lieb \cite{Lieb-83a}: for any $f_1,\ldots,f_M$ orthonormal in $L^2(\R^d)$ and for any non-negative coefficients $\nu_1,\ldots,\nu_M$, we have
\begin{equation}\label{eq:hom-Sob-M-fct}
 \norm{\sum_{j=1}^M \nu_j\left|(-\Delta)^{-s/2}f_j\right|^2}_{L^{\frac{d}{d-2s}}(\R^d)}\le C\left(\sup_j\nu_j\right)^{\frac{2s}{d}}\left(\sum_{j=1}^M\nu_j\right)^{\frac{d-2s}{d}}.
\end{equation}
Again, combining \eqref{eq:hom-Sob-1-fct} with the triangle inequality leads to the estimate
$$\norm{\sum_{j=1}^M \nu_j\left|(-\Delta)^{-s/2}f_j\right|^2}_{L^{\frac{d}{d-2s}}(\R^d)}\le \sum_{j=1}^M\nu_j\norm{(-\Delta)^{-s/2}f_j}_{L^{\frac{2d}{d-2s}}(\R^d)}^2\le C\sum_{j=1}^M\nu_j,$$
which is weaker than \eqref{eq:hom-Sob-M-fct}. Finally, a recent work by Frank, Lewin, Lieb, and Seiringer \cite{FraLewLieSei-14} generalizes the Strichartz inequality
\begin{equation}\label{eq:strichartz-1-fct}
  \norm{e^{it\Delta}f}_{L^p_tL^q_x(\R\times\R^d)}\le C\norm{f}_{L^2_x(\R^d)},\:d\ge1,\:p,q\ge2,\:\frac2p+\frac dq=\frac d2,\:(d,p,q)\neq(2,\ii,2),
\end{equation}
to orthonormal functions $f_1,\ldots,f_M$ in $L^2(\R^d)$ with complex coefficients $\nu_1,\ldots,\nu_M$:
\begin{equation}\label{eq:strichartz-M-fct}
 \norm{\sum_{j=1}^M\nu_j\left|e^{it\Delta}f_j\right|^2}_{L^{\frac p2}_tL^{\frac q2}_x(\R\times\R^d)}\le C\left(\sum_{j=1}^M|\nu_j|^\alpha\right)^{1/\alpha},\:1\le\alpha\le2q/(q+2),\:q\le2+\frac4d.
\end{equation}
These generalized Strichartz estimates were used in \cite{LewSab-13a,LewSab-13b} to study the nonlinear evolution of quantum systems with an infinite number of particles. 

Our motivation to prove inequalities of the form \eqref{eq:ES-orthonormal} is threefold. Harmonic analysis tools are widely used in nonlinear problems. With a perspective to study nonlinear many-body problems, it is thus natural to understand the model restriction problem in this many-body context. In a more concrete approach, it was noticed by Strichartz \cite{Strichartz-77} that the restriction problem for some quadratic surfaces is linked to space-time decay estimates for some evolution equations. As a consequence, we will see that \eqref{eq:ES-orthonormal} when $S$ is one of these quadratic surfaces also has an interpretation in terms of solutions to evolution equations. In particular, we provide a new proof of the Strichartz inequality \eqref{eq:strichartz-M-fct} which furthermore includes the \emph{full range of exponents} where it is valid. As in the original article of Strichartz, this also provides new Strichartz inequalities for orthonormal functions for the fractional Laplacian $(-\Delta)^{1/2}$ and for the pseudo-relativistic operator $(1-\Delta)^{1/2}$. Finally, we present a general principle which allows to obtain bounds of the kind \eqref{eq:ES-orthonormal}, not necessarily in the context of the extension operator. The advantage of this principle is that it allows to deduce \emph{automatically} bounds for orthonormal systems from bounds for a single function, if these latter were proved by a certain method based on complex interpolation. Let us now provide some insight about this principle. 

One of the reasons why the case $q=2$ of the restriction problem is better understood is that we can compose the maps $\cE_S$ and $\cR_S=(\cE_S)^*$. In particular, $\cE_S$ is bounded from $L^2(S)$ to $L^{p'}(\R^N)$ if and only if $T_S:=\cE_S(\cE_S)^*$ is bounded from $L^p(\R^N)$ to $L^{p'}(\R^N)$. Stein \cite{Stein-86} and Strichartz \cite{Strichartz-77} prove the boundedness of $T_S$ by introducing an analytic family of operators $(T_z)$ defined on a strip $a\le\text{Re}\, z\le b$ in the complex plane, such that $T_S=T_c$ for some $c\in(a,b)$. They prove that $T_z$ is a bounded operator from $L^2$ to $L^2$ on the line $\text{Re}\,z=b$ and from $L^1$ to $L^\ii$ on the line $\text{Re}\,z=a$. Using Stein's interpolation theorem \cite{Stein-56}, they deduce that $T_S=T_c$ is bounded from $L^p$ to $L^{p'}$ for some exponent $p$, which turns out to be the optimal one. Now notice that Hölder's inequality  implies that $T_S$ is bounded from $L^p(\R^N)$ to $L^{p'}(\R^N)$ if and only if for any $W_1,W_2\in L^{2p/(2-p)}(\R^N)$, the operator $W_1T_SW_2$ is bounded from $L^2(\R^N)$ to $L^2(\R^N)$ with the estimate
\begin{equation}\label{eq:bounded-WTSW}
 \norm{W_1T_SW_2}_{L^2(\R^N)\to L^2(\R^N)}\le C\norm{W_1}_{L^{2p/(2-p)}(\R^N)}\norm{W_2}_{L^{2p/(2-p)}(\R^N)},
\end{equation}
with $C>0$ independent of $W_1,W_2$. 

Our key contribution is to show that the operator $W_1T_SW_2$ is more than a mere bounded operator on $L^2$, namely that it belongs to a Schatten class. Recall that the Schatten class $\gS^\alpha(L^2(\R^N))$, $\alpha>0$, is defined as the space of all compact operators on $L^2(\R^N)$ such that the sequence of their singular values belongs to $\ell^\alpha$. The $\gS^\alpha$-norm of such an operator is then the $\ell^\alpha$-norm of its singular values. The estimate that we prove is 
\begin{equation}\label{eq:schatten-WTSW}
 \norm{W_1T_SW_2}_{\gS^\alpha(L^2(\R^N))}\le C\norm{W_1}_{L^{2p/(2-p)}(\R^N)}\norm{W_2}_{L^{2p/(2-p)}(\R^N)},
\end{equation}
with $C>0$ independent of $W_1,W_2$. Of course, \eqref{eq:schatten-WTSW} implies \eqref{eq:bounded-WTSW} and is hence stronger. By a well-known duality argument (see Lemma \ref{lemma:duality-principle} below), the bound \eqref{eq:schatten-WTSW} is equivalent to a bound of the kind \eqref{eq:ES-orthonormal} for systems of orthonormal functions. The estimate \eqref{eq:schatten-WTSW} follows from a general principle which states that, as soon as a linear operator $T$ belongs to an analytic family of operators of the same type as the one we described above, then the operator $W_1TW_2$ satisfies a bound of the form \eqref{eq:schatten-WTSW}. As we explained, under such assumptions on $T$, Stein's interpolation theorem would typically imply that $T$ is bounded from $L^p$ to $L^{p'}$. Our main input is to notice that $T$ actually satisfies a stronger Schatten bound \eqref{eq:schatten-WTSW}.

This general principle does not depend on the fact that $T_S$ can be decomposed as $T_S=\cE_S(\cE_S)^*$. This decomposition is only necessary to deduce \eqref{eq:ES-orthonormal} from \eqref{eq:schatten-WTSW}. However, the Schatten bound \eqref{eq:schatten-WTSW} has an interest by itself, even when it holds for an operator $T$ which is not of the form $T=AA^*$. An example of such an operator is given by the resolvent of the Laplacian on $\R^N$,
$$T=(-\Delta-z)^{-1},\qquad z\in\C\setminus[0,\ii).$$
In \cite{KenRuiSog-87}, Kenig, Ruiz, and Sogge prove the bound 
\begin{equation}\label{eq:keruso}
  \norm{(-\Delta-z)^{-1}}_{L^p(\R^N)\to L^{p'}(\R^N)}\le C|z|^{-1+N\left(\frac1p-\frac12\right)},
\end{equation}
for a range of exponents $p$ which depends on $N$. When $z$ varies away from a neighborhood of the origin, these bounds are uniform in $z$ and hence were labeled \emph{uniform Sobolev estimates}. In \cite{KenRuiSog-87}, the authors used them to prove unique continuation results for solutions to Schrödinger equations, but numerous other applications of these estimates were later found. For instance, they were part of Goldberg and Schlag's proof of the Limiting Absorption Principle in $L^p$ spaces \cite{GolSch-04}; see also \cite{IonSch-06}. In \cite{Frank-11}, they were used to control the size of eigenvalues of Schrödinger operators with complex-valued potentials. In this paper, we show that the boundedness statement contained in \eqref{eq:keruso} can be upgraded to a quantitative compactness statement. That is, we prove a Schatten bound of the type \eqref{eq:schatten-WTSW} for $T=(-\Delta-z)^{-1}$, using our general principle. As an application of this, we then deduce from it a version of the Limiting Absorption Principle in Schatten spaces. This will allow us to prove optimal bounds on the scattering matrix for Schr\"odinger operators. As another application, we will prove bounds for \emph{sums} of eigenvalues of Schr\"odinger operators with complex-valued potentials. These bounds control the accumulation of eigenvalues at points in $(0,\infty)$ in a significantly better way than previous bounds.

The organization of the article is the following, with the corresponding main results of each section:

\begin{itemize}[leftmargin=*]

 \item In Section \ref{sec:complex-interpolation}, we first explain our general principle (Proposition \ref{prop:complex-interpolation}) from which the Schatten bounds \eqref{eq:schatten-WTSW} follow. We then apply this principle to various situations.\\

 \item In Section \ref{sec:restriction}, we prove restriction theorems for systems of orthonormal functions (Theorem \ref{thm:restriction-orthonormal}). We also prove a result about the optimality of these estimates (Theorem \ref{thm:optimality-schatten-exponent}).\\

 \item In Section \ref{sec:strichartz}, we extend the Strichartz inequalities of \cite{FraLewLieSei-14} to the full range of exponents (Theorem \ref{thm:strichartz-orthonormal-general}). We furthermore prove new Strichartz inequalities for systems of orthonormal functions, for the operators $(-\Delta)^{1/2}$ (Theorem \ref{thm:strichartz-orthonormal-fractional}) and $(1-\Delta)^{1/2}$ (Theorem \ref{thm:strichartz-orthonormal-pseudo-relativistic}).\\

 \item In Section \ref{sec:sobolev}, we prove uniform Sobolev estimates in Schatten spaces (Theorem \ref{thm:uniform-sobolev-schatten}), and apply them to prove a Limiting Absorption Principle in Schatten spaces (Theorem \ref{thm:LAP}).\\

 \item In Section \ref{sec:hartree}, we give an application of Strichartz estimates to well-posedness of the non-linear Hartree equation in Schatten spaces (Theorem \ref{thm:gwp-hartree}).\\

 \item In Section \ref{sec:LT}, we give another application of uniform Sobolev estimates to Lieb--Thirring inequalities for Schrödinger operators with complex-valued potentials (Theorem \ref{thm:LT-complex}).\\

 \item In Section \ref{sec:SCATMAT}, we finally give an application of the Limiting Absorption Principle to Schatten estimates of the Scattering Matrix (Theorem \ref{thm:scatmat}).
    
 \end{itemize}

\section{A complex interpolation estimate in Schatten spaces}\label{sec:complex-interpolation}

In this section, we explain how to obtain Schatten bounds of the form \eqref{eq:schatten-WTSW} by a complex interpolation method. The advantage of this result is that it requires assumptions that are naturally proved when one wants to show that a given operator is bounded from $L^p(\R^N)$ to $L^{p'}(\R^N)$ by a complex interpolation method. Hence, it provides an automatic way to upgrade this $L^p\to L^{p'}$ bound into a stronger Schatten bound. 

Let us first recall that a family of operators $(T_z)$ on $\R^N$ defined on a strip $a\le\text{Re}\,z\le b$ in the complex plane ($a<b$) is analytic in the sense of Stein \cite{Stein-56} if for all simple functions $f,g$ on $\R^N$ (that is, functions that take a finite number of nonzero values on sets of finite measure in $\R^N)$, the map $z\mapsto\langle g,T_z f\rangle$ is analytic in $a<{\rm Re}\,z<b$, continuous in $a\le{\rm Re}\,z\le b$, and if 
$$\sup_{a\le x \le b}|\langle g,T_{x+is}f\rangle|\le C(s),$$
for some $C(s)$ with at most a (double) exponential growth in $s$. 

We also recall the definition of Schatten spaces, see, e.g., \cite{Simon-79}. Let $\gH$ be a complex Hilbert space. For any compact operator $T$ on $\gH$, the non-zero eigenvalues of $\sqrt{T^*T}$ are called the singular values of $T$. They form an (at most) countable set that we denote by $(\mu_n(T))_{n\in\N}$. For $\alpha>0$, the Schatten space $\gS^\alpha(\gH)$ is defined as the space of all compact operators $T$ on $\gH$ such that $\sum_{n\in\N}\mu_n(T)^\alpha<\ii$. When $\alpha\ge1$, it is a Banach space endowed with the norm
$$\norm{T}_{\gS^\alpha(\gH)}:=\left(\sum_{n\in\N}\mu_n(T)^\alpha\right)^{1/\alpha}.$$
Our result is the following.

\begin{proposition}\label{prop:complex-interpolation}
  Let $(T_z)$ be an analytic family of operators on $\R^N$ in the sense of Stein defined on the strip $-\lambda_0\le{\rm Re}\,z\le0$ for some $\lambda_0>1$. Assume that we have the bounds
  \begin{equation}
    \norm{T_{is}}_{L^2\to L^2}\le M_0 e^{a|s|},\quad\norm{T_{-\lambda_0+is}}_{L^1\to L^\ii}\le M_1e^{b|s|},\quad\forall s\in\R,
  \end{equation}
  for some $a,b\ge0$ and for some $M_0,M_1\ge0$. Then, for all $W_1,W_2\in L^{2\lambda_0}(\R^N,\C)$, the operator $W_1T_{-1}W_2$ belongs to $\gS^{2\lambda_0}(L^2(\R^N))$ and we have the estimate
  \begin{equation}\label{eq:schatten-bound-complex-interpolation}
   \norm{W_1T_{-1}W_2}_{\gS^{2\lambda_0}(L^2(\R^N))}\le M_0^{1-\frac{1}{\lambda_0}}M_1^{\frac{1}{\lambda_0}}\norm{W_1}_{L^{2\lambda_0}(\R^N)}\norm{W_2}_{L^{2\lambda_0}(\R^N)}.
  \end{equation} 
\end{proposition}

\begin{proof}
 Let $W_1,W_2$ be non-negative, simple functions, and define the family of operators
 $$S_z:=W_1^{-z}T_zW_2^{-z}.$$
 The family $(S_z)$ is still analytic in the sense of Stein in the strip $-\lambda_0\le{\rm Re}\,z\le0$, and satisfies $S_{-1}=W_1T_{-1}W_2$. For all $s\in\R$ we have the first bound
 $$\norm{S_{is}}_{L^2\to L^2}\le \norm{W_1^{-is}}_{L^\ii}\norm{T_{is}}_{L^2\to L^2}\norm{W_2^{-is}}_{L^\ii}\le M_0e^{a|s|}.$$
 By the Dunford--Pettis theorem \cite[Thm. 2.2.5]{DunPet-40}, the operator $T_{-\lambda_0+is}$ has an integral kernel $T_{-\lambda_0+is}(x,y)$ satisfying
 $$\norm{T_{-\lambda_0+is}(\cdot,\cdot)}_{L^\ii(\R^N\times\R^N)}=\norm{T_{-\lambda_0+is}}_{L^1(\R^N)\to L^\ii(\R^N)}\le M_1e^{b|s|},\quad\forall s\in\R.$$
 Hence, we deduce the Hilbert--Schmidt bound
 \begin{align*}
    \norm{S_{-\lambda_0+is}}_{\gS^2}^2 &=\int_{\R^N}\int_{\R^N}W_1(x)^{2\lambda_0}\left|T_{-\lambda_0+is}(x,y)\right|^2W_2(y)^{2\lambda_0}\,dxdy\\
    &\le M_1^2e^{2b|s|}\norm{W_1}_{L^{2\lambda_0}(\R^N)}^{2\lambda_0}\norm{W_2}_{L^{2\lambda_0}(\R^N)}^{\lambda_0}.
 \end{align*}
 By \cite[Thm. 2.9]{Simon-79}, we deduce that $S_{-1}$ belongs to $\gS^{2\lambda_0}(L^2(\R^N))$ with 
 $$\norm{S_{-1}}_{\gS^{2\lambda_0}(L^2(\R^N))}\le M_0^{1-\frac{1}{\lambda_0}}M_1^{\frac{1}{\lambda_0}}\norm{W_1}_{L^{2\lambda_0}(\R^N)}\norm{W_2}_{L^{2\lambda_0}(\R^N)}.$$
 Hence, we have proved \eqref{eq:schatten-bound-complex-interpolation} for $W_1,W_2$ non-negative and simple. The non-negativity assumption can be removed by writing $W_j=e^{i\phi_j}|W_j|$ and estimating
 $$\norm{W_1T_{-1}W_2}_{\gS^{2\lambda_0}(L^2(\R^N))}\le \norm{e^{i\phi_1}}_{L^2\to L^2}\norm{|W_1|T_{-1}|W_2|}_{\gS^{2\lambda_0}}\norm{e^{i\phi_2}}_{L^2\to L^2}\le \norm{|W_1|T_{-1}|W_2|}_{\gS^{2\lambda_0}},$$
 and the simplicity assumption is removed by density.
\end{proof}

\begin{remark}
 The previous proof shows that the conclusion of Proposition \ref{prop:complex-interpolation} also holds when $\lambda_0=1$: in this case, there is even no interpolation to perform. 
\end{remark}

When furthermore we can decompose $T_{-1}=AA^*$, we deduce from \eqref{eq:schatten-bound-complex-interpolation} a corresponding result for systems of orthonormal functions.

\begin{lemma}[Duality principle]\label{lemma:duality-principle}
 Let $\gH$ be a separable Hilbert space. Assume that $A$ is a bounded operator from $\gH$ to $L^{p'}(\R^N)$ for some $1\le p\le2$ and let $\alpha\ge 1$. Then the following are equivalent.
\begin{enumerate}
\item[(i)] There is a constant $C>0$ such that
\begin{equation}\label{eq:schatten-WA}
  \norm{WAA^*\bar{W}}_{\gS^\alpha(L^2(\R^N))}\le C\norm{W}_{L^{2p/(2-p)}(\R^N)}^2,\quad\forall W\in L^{2p/(2-p)}(\R^N,\C).
 \end{equation}
\item[(ii)] There is a constant $C'>0$ such that for any orthonormal system $(f_j)_{j\in J}$ in $\gH$ and any sequence $(\nu_j)_{j\in J}\subset\C$,
 \begin{equation}\label{eq:est-orthonormal-H}
  \norm{\sum_{j\in J}\nu_j\left|Af_j\right|^2}_{L^{p'/2}(\R^N)}\le C' \left(\sum_{j\in J}|\nu_j|^{\alpha'}\right)^{1/\alpha'} \,.
 \end{equation}
\end{enumerate}
Moreover, the values of the optimal constants $C$ and $C'$ coincide.
\end{lemma}

\begin{proof}
 First, notice that (ii) is equivalent to (ii'), which is the same as (ii) but with the additional restriction that all $\nu_j\geq 0$. Indeed, assuming (ii') and taking $(\nu_j)\subset\C$, one has by the triangle inequality in $\C$
 $$\norm{\sum_{j\in J}\nu_j\left|Af_j\right|^2}_{L^{p'/2}(\R^N)}\le\norm{\sum_{j\in J}|\nu_j|\left|Af_j\right|^2}_{L^{p'/2}(\R^N)} \le C' \left(\sum_{j\in J}|\nu_j|^{\alpha'}\right)^{1/\alpha'} \,,$$
which is (ii) (with the same constant $C'$). We thus show that (i) implies (ii'). Let $(f_j)_{j\in J}$ an orthonormal system in $\gH$ and $(\nu_j)_{j\in J}\subset\R_+$. We define an operator $\gamma$ on $\gH$ with eigenfunctions $(f_j)$ and corresponding eigenvalues $(\nu_j)$. In Dirac's notation, we have
 $$\gamma=\sum_{j\in J}\nu_j|f_j\rangle\langle f_j|,$$
 where $|f_j\rangle\langle f_j|$ denotes the orthogonal projection on $\C f_j\subset\gH$. Estimate \eqref{eq:schatten-WA} is equivalent to 
 \begin{equation}\label{eq:schatten-W2A}
  \norm{A^*|W|^2A}_{\gS^\alpha(\gH)}\le C\norm{W}_{L^{2p/(2-p)}(\R^N)}^2,\quad\forall W\in L^{2p/(2-p)}(\R^N,\C).
 \end{equation}
 Using \eqref{eq:schatten-W2A} and Hölder's inequality in Schatten spaces \cite[Thm. 2.8]{Simon-79}, we deduce that 
 \begin{align*}
    \tr_{L^2(\R^N)}(WA\gamma(WA)^*)=\tr_{\gH}(\gamma A^*|W|^2 A)&\le C\norm{\gamma}_{\gS^{\alpha'}(\gH)}\norm{W}_{L^{2p/(2-p)}(\R^N)}^2\\
    &=C\left(\sum_{j\in J}\nu_j^{\alpha'}\right)^{1/\alpha'}\norm{W}_{L^{2p/(2-p)}(\R^N)}^2.
 \end{align*}
Since we have the identity
$$\tr_{L^2(\R^N)}(WA\gamma(WA)^*)=\int_{\R^N}\left(\sum_{j\in J}\nu_j|(Af_j)(x)|^2\right)|W(x)|^2\,dx,$$
we infer that for all $V\in L^{p/(2-p)}(\R^N)$ with $V\ge0$, 
$$\int_{\R^N}\left(\sum_{j\in J}\nu_j|(Af_j)(x)|^2\right)V(x)\,dx\le C\left(\sum_{j\in J}\nu_j^{\alpha'}\right)^{1/\alpha'}\norm{V}_{L^{p/(2-p)}(\R^N)}.$$
The duality principle for $L^p$-spaces (or choosing $V\equiv1$ when $p=2$) leads to \eqref{eq:est-orthonormal-H}, since $(p/(2-p))'=p'/2$. Thus, (i) implies (ii'). The proof that (ii') implies (i) is similar and will be omitted.  
\end{proof}

\begin{remark}\label{rk:density}
 The previous proof shows that \eqref{eq:schatten-WA} and \eqref{eq:est-orthonormal-H} are equivalent to the following bound: for any $\gamma\in\gS^{\alpha'}(\gH)$, we have 
 \begin{equation}\label{eq:est-rho-Agamma}
  \norm{\rho_{A\gamma A^*}}_{L^{p'/2}(\R^N)}\le C\norm{\gamma}_{\gS^{\alpha'}(\gH)},
 \end{equation}
 with $C>0$ independent of $\gamma$, and where $\rho_{A\gamma A^*}$ is the \emph{density} of the operator $A\gamma A^*$. It is defined for any finite-rank $\gamma$ by duality,
 $$\int_{\R^N}\rho_{A\gamma A^*}(x)V(x)\,dx:=\tr_{\gH}(\gamma A^* V A),$$
 and extended to all $\gamma\in\gS^{\alpha'}(\gH)$ using the density of finite-rank operators in $\gS^{\alpha'}(\gH)$ and the estimate \eqref{eq:est-rho-Agamma} valid for all finite-rank $\gamma$. 
\end{remark}

To illustrate this duality principle, let us consider the case of Young's inequality. The underlying bounded operator $A$ from $\gH=L^2(\R^N)$ to $L^{p'}(\R^N)$ is $Af=g*f$ for some fixed $g\in L^{2p'/(2+p')}(\R^N)$. Then, the corresponding Schatten bound \eqref{eq:schatten-WA} is the Kato--Seiler--Simon inequality \cite[Thm. 4.1]{Simon-79},
\begin{equation}\label{eq:KSS}
  \norm{W|\hat{g}(-i\nabla)|^2\bar{W}}_{\gS^{p/(2-p)}(L^2(\R^N))}\le(2\pi)^{N(1-2/p)}\norm{W}_{L^{2p/(2-p)}(\R^N)}^2\norm{\hat{g}}_{L^{2p/(2-p)}}^2.
\end{equation}
We note that our proof of Proposition \ref{prop:complex-interpolation} is based on complex interpolation much like the proof of \eqref{eq:KSS} in \cite{Simon-79}. Together with Lemma \ref{lemma:duality-principle}, \eqref{eq:KSS} implies the following Young inequality for systems of orthonormal functions, which we have not encountered in the literature in this form. A version in terms of densities, however, is given by \cite[Lemma 1]{LewSab-13a}. 

\begin{theorem}[Young inequality for orthonormal functions]
 Let $N\ge1$, $1\le p\le 2$, and $g\in L^{2p'/(2+p')}(\R^N)$. Then, for any (possibly infinite) orthonormal system $(f_j)$ in $L^2(\R^N)$ and for any $(\nu_j)\subset\C$, we have
\begin{equation}\label{eq:young-orthonormal}
  \norm{\sum_j \nu_j\left|g*f_j\right|^2}_{L^{p'/2}(\R^N)}\le (2\pi)^{\frac{2N}{p'}}\norm{\hat{g}}_{L^{2p/(2-p)}}^2\left(\sum_j|\nu_j|^{\frac{p'}{2}}\right)^{\frac{2}{p'}}.
 \end{equation}
\end{theorem}

\begin{remark}
 By the Hausdorff--Young inequality, $\norm{\hat{g}}_{L^{2p/(2-p)}}$ is controlled by $\norm{g}_{L^{2p'/(2+p')}}$. For this reason, \eqref{eq:young-orthonormal}, for a single function $f$, is somewhat stronger than Young's inequality.
\end{remark}

\section{Restriction theorems}\label{sec:restriction}

\subsection{Restriction theorems for orthonormal functions}\label{sec:results-restriction}

As explained in the introduction, we consider the same surfaces as Stein \cite{Stein-86} and Strichartz \cite{Strichartz-77}. The surfaces considered by Stein are smooth, compact surfaces embedded in $\R^N$ ($N\ge2$) with non-zero Gauss curvature, endowed with their $(N-1)$-dimensional Lebesgue measure that we denote by $d\sigma$. The quadratic surfaces considered by Strichartz are split into three categories: 

\begin{itemize}[leftmargin=*]
 \item \textbf{Case I:} $S=\{\xi\in\R^N,\,\xi_N=\xi_1^2+\cdots+\xi_a^2-\xi_{a+1}^2-\cdots-\xi_{N-1}^2\}$ where $a=0,\ldots,N-1$. The model case of a surface of this kind is the paraboloid ($a=0,N-1$). In this case, the measure is chosen to be $(1+4\xi_1^2+\cdots+4\xi_{N-1}^2)^{-1/2}d\sigma(\xi)$, where $d\sigma$ is the induced $(N-1)$-dimensional Lebesgue measure on $S$.

 \item \textbf{Case II:} $S=\{\xi\in\R^N,\,\xi_1^2+\cdots+\xi_a^2-\xi_{a+1}^2-\cdots-\xi_N^2=0\}$, where $a=1,\ldots,N-1$. The model case here is the cone ($a=N-1$). The measure is chosen to be $(2|\xi|)^{-1}d\sigma(\xi)$.

 \item \textbf{Case III:} $S=\{\xi\in\R^N,\,\xi_1^2+\cdots+\xi_a^2-\xi_{a+1}^2-\cdots-\xi_N^2=-1\}$, where $a=0,\ldots,N-1$. There are two model cases here: the sphere ($a=0$) and the two-sheeted hyperboloid ($a=N-1$). The measure is chosen to be $(2|\xi|)^{-1}d\sigma(\xi)$.
\end{itemize}
 
Notice that in the case of quadratic surfaces, the measure is not the usual surface measure $d\sigma$. Writing $S$ as $S=\{\xi:\,R(\xi)=0\}$ where $R$ is the degree two polynomial appearing in the definition of $S$, we see that the chosen measure is simply $|\nabla R(\xi)|^{-1}d\xi$.

In any of these two cases, the extension operator $\cE_S$ is defined by \eqref{eq:def-ES}, and we denote $T_S=\cE_S(\cE_S)^*$. Our Schatten bounds on $T_S$ are the following. 

\begin{theorem}[Schatten properties of extension maps---compact case]\label{thm:schatten-WTSW}
 Let $N\ge2$, let $S\subset\R^N$ be a smooth, compact surface with non-zero Gauss curvature and let $1\leq q\leq (N+1)/2$. Then the inequality
 \begin{equation}\label{eq:schatten-WTSW-thm}
  \norm{W_1T_SW_2}_{\gS^{(N-1)q/(N-q)}(L^2(\R^N))}\le C\norm{W_1}_{L^{2q}(\R^N)}\norm{W_2}_{L^{2q}(\R^N)}
 \end{equation}
  holds for all $W_1,W_2$ with a constant $C>0$ independent of $W_1,W_2$.
\end{theorem}

We shall see later (Theorem \ref{thm:optimality-schatten-exponent}) that for any $1\leq q\leq (N+1)/2$ the Schatten exponent $(N-1)q/(N-q)$ on the left side of \eqref{eq:schatten-WTSW-thm} is optimal (that is, smallest possible). Also the condition $2q\leq N+1$ on the Lebesgue space of $W_1$ and $W_2$ is best possible, since \eqref{eq:schatten-WTSW-thm} fails for $q>(N+1)/2$ even with the operator norm on the left side. This follows from Knapp's argument; see, e.g., \cite{Strichartz-77}.

\begin{theorem}[Schatten properties of extension maps---quadratic case]\label{thm:schatten-WTSW2}
 Let $N\ge2$ and let $S\subset\R^N$ be a quadratic surface. Then the inequality
 \begin{equation}\label{eq:schatten-WTSW-thm2}
  \norm{W_1T_SW_2}_{\gS^{2q}(L^2(\R^N))}\le C\norm{W_1}_{L^{2q}(\R^N)}\norm{W_2}_{L^{2q}(\R^N)}
 \end{equation}
  holds for all $W_1,W_2$ with a constant $C>0$ independent of $W_1,W_2$, under the following assumptions on the exponent $q$:
  \begin{itemize}
  \item Case I: $q=(N+1)/2$;
  \item Case II: $q=N/2$;
  \item Case III: 

  (i) $a=0$ and $1\le q\le (N+1)/2$;
  
  (ii) $a\neq0$, $N\ge3$, and $N/2\le q\le (N+1)/2$;
  
  (iii) $a=1$, $N=2$, and $1<p\le 3/2$.  
 \end{itemize}  
\end{theorem}

These theorems have equivalent formulations as restriction estimates for systems of orthonormal functions, which we present next. Combining Theorems \ref{thm:schatten-WTSW} and \ref{thm:schatten-WTSW2} with Lemma \ref{lemma:duality-principle} for $\gH:=L^2(S,d\sigma)$, we immediately deduce the following results.

\begin{theorem}[Restriction estimates for orthonormal functions---compact case]\label{thm:restriction-orthonormal}
 Let $N\ge2$ and $S\subset\R^N$ a smooth, compact surface with non-zero Gauss curvature. Then, for any (possibly infinite) orthonormal system $(f_j)$ in $L^2(S,d\sigma)$ and for any $(\nu_j)\subset\C$, we have
\begin{equation}\label{eq:restriction-orthonormal}
  \norm{\sum_j \nu_j\left|\cE_S f_j\right|^2}_{L^{q'}(\R^N)}\le C\left(\sum_j|\nu_j|^{\frac{(N-1)q}{N(q-1)}}\right)^{\frac{N(q-1)}{(N-1)q}},
 \end{equation}
 with $C>0$ independent of $(\nu_j)$ and $(f_j)$. Here $q'=q/(q-1)$, where $q$ satisfies the same assumptions as in Theorem \ref{thm:schatten-WTSW}.
\end{theorem}

\begin{theorem}[Restriction estimates for orthonormal functions---quadratic case]\label{thm:restriction-orthonormal2}
 Let $N\ge2$ and let $S\subset\R^N$ be a quadratic surface. Then, for any (possibly infinite) orthonormal system $(f_j)$ in $L^2(S,d\sigma)$ and for any $(\nu_j)\subset\C$, we have
\begin{equation}\label{eq:restriction-orthonormal2}
  \norm{\sum_j \nu_j\left|\cE_S f_j\right|^2}_{L^{q'}(\R^N)}\le C\left(\sum_j|\nu_j|^{\frac{2q}{2q-1}}\right)^{1-\frac{1}{2q}},
 \end{equation}
 with $C>0$ independent of $(\nu_j)$ and $(f_j)$. The exponent $q$ satisfies the same assumptions as in Theorem \ref{thm:schatten-WTSW2}, according to the type of $S$.
\end{theorem}

Furthermore, according to Remark \ref{rk:density}, inequalities \eqref{eq:restriction-orthonormal} and \eqref{eq:restriction-orthonormal2} can be rewritten as
\begin{equation}\label{eq:extension-rho-schatten}
  \norm{\rho_{\cE_S\gamma(\cE_S)^*}}_{L^{q'}(\R^N)}\le C\norm{\gamma}_{\gS^{(N-1)q/(N(q-1))}(L^2(S,d\sigma))} \,,
 \end{equation}
and
\begin{equation}\label{eq:extension-rho-schatten2}
  \norm{\rho_{\cE_S\gamma(\cE_S)^*}}_{L^{q'}(\R^N)}\le C\norm{\gamma}_{\gS^{2q/(2q-1)}(L^2(S,d\sigma))},
 \end{equation}
for any $\gamma\in\gS^{(N-1)q/(N(q-1))}(L^2(S,d\sigma))$ and any $\gamma\in\gS^{2q/(2q-1)}(L^2(S,d\sigma))$, respectively. Later on the form \eqref{eq:extension-rho-schatten} will be convenient to prove optimality of the corresponding Schatten exponent.

\subsection{Proof of Theorem \ref{thm:schatten-WTSW}}

Let us first consider the case $S$ compact with non-zero Gauss curvature. Thus, let $N\ge2$ and $S\subset\R^N$ a compact hypersurface with non-zero Gauss curvature, which is endowed with its $(N-1)$-dimensional Lebesgue measure $d\sigma$. 

\noindent\emph{Step 1.} We shall prove that
\begin{equation}\label{eq:schatten-WTSW-thmproof}
  \norm{W_1T_SW_2}_{\gS^{2p/(2-p)}(L^2(\R^N))}\le C\norm{W_1}_{L^{2p/(2-p)}(\R^N)}\norm{W_2}_{L^{2p/(2-p)}(\R^N)}
 \end{equation}
for any $1\leq p \leq 2(N+1)/(N+3)$. (Later we will use this only for $p=2(N+1)/(N+3)$, but it is instructive the prove the more general inequality.) The operator $T_S$ acting on functions on $\R^N$ is a convolution operator: $T_Sf=K_S*f$ for all $f$, where $K_S$ is the function
$$K_S(x)=\frac{1}{(2\pi)^N}\int_S e^{ix\cdot\xi}d\sigma(\xi),\quad\forall x\in\R^N.$$
Using a smooth and finite partition of unity, $1=\sum_\ell \psi_\ell$ on $S$, the operator $T_S$ can be decomposed as 
$T_S=\sum_\ell T_\ell$, where $T_\ell$ is the convolution operator by the function
$$K_\ell(x)=\frac{1}{(2\pi)^N}\int_S e^{ix\cdot\xi}\psi_\ell(\xi)d\sigma(\xi),\quad\forall x\in\R^N.$$
The partition of unity is chosen in the following fashion. We assume that on the interior of the support of $\psi_\ell$, the surface $S$ is the graph of a smooth and compactly supported function $\phi:\R^{N-1}\to\R$, so that (possibly after a rotation),
$$K_\ell(x)=\frac{1}{(2\pi)^N}\int_{\R^{N-1}}e^{ix\cdot(\xi',\phi(\xi'))}\psi_\ell(\xi',\phi(\xi'))(1+|\nabla\phi(\xi')|^2)^{1/2}d\xi',\quad\forall x\in\R^N.$$
To prove \eqref{eq:schatten-WTSW-thmproof}, it is then enough to show the estimate
$$\norm{W_1T_\ell W_2}_{\gS^{2p/(2-p)}(L^2(\R^N))}\le C\norm{W_1}_{L^{2p/(2-p)}(\R^N)}\norm{W_2}_{L^{2p/(2-p)}(\R^N)}$$
for each $\ell$. Hence, from now on we drop the index $\ell$ and write $(T,K)$ instead of $(T_\ell,K_\ell)$. To prove this Schatten estimate, we use Proposition \ref{prop:complex-interpolation} by defining the same analytic family $(T_z)$ of operators as in \cite{Stein-86}. More precisely, let $T_z$ be the convolution operator with the function $K_z$ defined as
$$K_z(x):=\zeta_z(x_N)K(x),\quad\forall x\in\R^N,$$
where
$$\zeta_z(y):=\frac{1}{(2\pi)^N}\frac{e^{(z+1)^2}}{\Gamma(z+1)}\int_0^\ii e^{ity}t^z\eta(t)\,dt,\quad\forall y\in\R,$$
where $\eta$ is a smooth an compactly supported function on $\R$ such that $\eta\equiv1$ on a neighborhood of the origin. As explained in \cite[Ch. IX, Sec. 1.2.3]{Stein-93}, the family $(T_z)$ is an analytic family of operators in the strip $-\lambda_0\le\text{Re}\,z\le0$, with $1\le\lambda_0\le(N+1)/2$, which satisfies the estimate
$$\norm{T_{is}}_{L^2\to L^2}+\norm{T_{-\lambda_0+is}}_{L^1\to L^\ii}\le C(s),$$
for all $s\in\R$ and for some $C(s)$ growing exponentially in $s$. By Proposition \ref{prop:complex-interpolation} and the identity $T=T_{-1}$, we obtain \eqref{eq:schatten-WTSW-thmproof}.

\noindent\emph{Step 2.} In order to complete the proof of the theorem, we recall that $W_1 T_S W_2 = (W_1 \mathcal E_S)(\overline{W_2} \mathcal E_S)^*$. The operator $W \mathcal E_S$ acts from $L^2(S)$ to $L^2(\R^N)$ as an integral operator with integral kernel $(2\pi)^{-N/2}W(x) e^{i\xi\cdot x}$, where $\xi\in S$, $x\in\R^N$. Since $S$ is compact, it has finite surface measure $\sigma(S)$. Therefore, if $W\in L^2(\R^N)$, then $W \mathcal E_S$ is Hilbert--Schmidt with
$$
\norm{W \mathcal E_S}_{\gS^{2}(L^2(S),L^2(\R^N))}^2 = (2\pi)^{-N} \sigma(S)\norm{W}_{L^{2}(\R^N)}^2 \,.
$$
Thus, by H\"older's inequality for trace ideals,
$$
\norm{W_1T_SW_2}_{\gS^{1}(L^2(\R^N))}\le (2\pi)^{-N} \sigma(S) \norm{W_1}_{L^{2}(\R^N)}\norm{W_2}_{L^{2}(\R^N)} \,.
$$
On the other hand, in step 1 (with $p=2(N+1)/(N+3)$) we have shown that
$$
\norm{W_1T_SW_2}_{\gS^{N+1}(L^2(\R^N))}\le C\norm{W_1}_{L^{N+1}(\R^N)}\norm{W_2}_{L^{N+1}(\R^N)} \,.
$$
By complex interpolation between these two bounds \cite[Thm. 2.9]{Simon-79} we obtain the assertion of the theorem. \qed

\subsection{Proof of Theorem \ref{thm:schatten-WTSW2}}

All three kinds of quadratic surfaces considered by Strichartz \cite{Strichartz-77} are of the form $\{\xi\,:\,R(\xi)=r\}$ for some degree two polynomial $R$ and some $r\in\R$. Strichartz introduces the family of tempered  distributions $(G_z)_{z\in\C}$ on $\R^N$ as
 $$\langle G_z,\phi\rangle:=g(z)\int_{\R^N}(R(\xi)-r)_+^z\phi(\xi)d\xi:=g(z)\int_\R(a-r)_+^z\left(\int_{S_a}\phi(\xi)d\mu_a(\xi)\right)\,da,$$
 where $S_a=\{x\,:R(x)=a\}$ and $d\mu_a(\xi)=|\nabla R(\xi)|^{-1/2}d\sigma_a(\xi)$, with $d\sigma_a(\xi)$ the $(N-1)$-dimensional Lebesgue measure on $S_a$. The function $g(z)$ has adequate properties according to the type of the surface considered, but in all cases it has a simple zero at $z=-1$ to ensure that $G_{-1}\equiv\delta_S$. The family of operators $T_z$ is then defined as Fourier multipliers by $G_z$, that is 
 $$T_zf(x)=\langle G_z,\hat{f}(\xi)e^{i\xi\cdot x}\rangle,\quad\forall x\in\R^N.$$
 Strichartz then shows the bounds
 $$\norm{T_{is}}_{L^2\to L^2}+\norm{T_{-\lambda_0+is}}_{L^1\to L^\ii}\le C(s),$$
 for all $s\in\R$ and for $C(s)$ growing exponentially, with $\lambda_0=(N+1)/2$ (Case I), $\lambda_0=N/2$ (Case II), $\lambda_0\ge(N+1)/2$ (Case III(i)), $N/2\le\lambda_0\le(N+1)/2$ (Case III(ii)), and $1<\lambda_0\le3/2$ (Case III(iii)). Together with Proposition \ref{prop:complex-interpolation}, this shows \eqref{eq:schatten-WTSW-thm} in the quadratic case and the proof of Theorem \ref{thm:schatten-WTSW} is over. \qed

\begin{remark}
 Case I of a quadratic surface can also be deduced from the compact case via scaling, as in \cite[Sec. VIII.5.16]{Stein-93}.
\end{remark}

\subsection{Optimality of the Schatten exponent in the compact case}

We now prove that the Schatten space $\gS^{(N-1)q/(N-q)}(L^2(\R^N))$ in Theorem \ref{thm:schatten-WTSW} is optimal. By this we mean that the inequality fails if this space is replaced by $\gS^{s}(L^2(\R^N))$ for some $s<(N-1)q/(N-q)$. According to our duality principle, Lemma \ref{lemma:duality-principle}, this is equivalent to proving that the Schatten space $\gS^{(N-1)q/(N(q-1))}(L^2(S))$ in \eqref{eq:extension-rho-schatten} is optimal. Now optimality means that the inequality fails if this space is replaced by $\gS^{r}(L^2(S))$ for some $r>(N-1)q/(N(q-1))$. This is the content of the following theorem.

\begin{theorem}[Optimality of the Schatten exponent]\label{thm:optimality-schatten-exponent}
 Let $N\ge2$, let $S\subset\R^N$ be a smooth surface with non-zero Gauss curvature and let $1\le q\le (N+1)/2$. Then, for any $r>\frac{(N-1)q}{N(q-1)}$, we have 
 \begin{equation}
  \sup_{\substack{\gamma\in\gS^r(L^2(S)),\\ \gamma\neq0}}\frac{\norm{\rho_{\cE_S\gamma(\cE_S)^*}}_{L^{q'}(\R^N)}}{\norm{\gamma}_{\gS^r(L^2(S))}}=+\ii.
 \end{equation}
\end{theorem}

\begin{proof}
 Let $h>0$. We construct a trial operator $\gamma_h$ on $L^2(S)$, by defining its integral kernel:
 $$\gamma_h(\omega,\omega')=\int_{\R^N}\1(k^2\le h^{-2})e^{ik\cdot(\omega-\omega')}\,dk,\,\,\forall(\omega,\omega')\in S\times S.$$
 Let $f\in L^2(S)$. Using the Agmon--H\"ormander bound \cite{AgmHor-76}, \cite[Thm. 4.2]{Ruiz-harmonic}, we have
 $$\langle f,\gamma_h f\rangle=\int_{\R^N}\1(k^2\le h^{-2})\left|\int_S f(\omega)e^{-ik\cdot\omega}\,d\sigma(\omega)\right|^2\,dk\le Ch^{-1}\int_S|f(\omega)|^2\,d\sigma(\omega).$$
 This shows that $\gamma_h$ is a non-negative bounded operator on $L^2(S)$ with 
 $$\norm{\gamma_h}_{L^2(S)\to L^2(S)}\le Ch^{-1}.$$
 We also compute its trace norm,
 $$\norm{\gamma_h}_{\gS^1(L^2(S))}=\tr\gamma_h=\int_S\gamma_h(\omega,\omega)\,d\omega=Ch^{-N}.$$
 By H\"older's inequality in Schatten spaces, we deduce that $\gamma_h\in\gS^r$ for all $1\le r\le+\ii$ and that 
 \begin{equation}\label{eq:schatten-r-trial}
    \norm{\gamma_h}_{\gS^r(L^2(S))}\le C(h^{-1})^{\frac{N+r-1}{r}}.
 \end{equation}
 Let us compute the left side of \eqref{eq:extension-rho-schatten} for $\gamma_h$. We have 
 $$\rho_{\cE_S\gamma_h (\cE_S)^*}(\xi)=\int_{\R^N}\1(k^2\le h^{-2})\left|\int_S e^{-i\omega\cdot \xi}e^{ik\cdot\omega}\,d\sigma(\omega)\right|^2=\int_{\R^N}\1(k^2\le h^{-2})\left|\hat{d\sigma}(\xi-k)\right|^2\,dk,$$
 for all $\xi\in\R^N$. First, let us use the lower bound 
 \begin{align*}
    \norm{\rho_{\cE_S\gamma_h (\cE_S)^*}}_{L^{q'}(\R^N)} &\ge\left(\int_{|\xi|\le h^{-1}R}\rho_{\cE_S\gamma_h (\cE_S)^*}(\xi)^{q'}\,d\xi\right)^{1/q'} \\
    &=h^{-N/q'}\left(\int_{|\xi|\le R}\rho_{\cE_S\gamma_h (\cE_S)^*}(h^{-1}\xi)^{q'}\,d\xi\right)^{1/q'},
 \end{align*}
 for some $R>0$ to be chosen later on. Next, we use a lower bound on $\hat{d\sigma}$ which can be found, for instance, in \cite[p. 51]{Sogge-book}: there exists a non-empty open cone $\Gamma\subset\R^N$ such that for all $k\in\Gamma$ with $|k|\ge R'$ for some $R'>0$ large enough, we have
 \begin{equation}\label{eq:bound-fourier-measure}
    |\hat{d\sigma}(k)|\ge\frac{C}{(1+|k|)^{\frac{N-1}{2}}}.
 \end{equation}
 Notice that in the case of the sphere, \eqref{eq:bound-fourier-measure} can be proved directly using that $\hat{d\sigma}$ is explicitly expressed in terms of a Bessel function. In any case, this implies that, if $h\le1$,
 \begin{align*}
    \rho_{E\gamma_h E^*}(h^{-1}\xi) &\ge \int_{\substack{k\in\Gamma \\ |k|\ge h^{-1}R'}}\1(|k-h^{-1}\xi|\le h^{-1})\left|\hat{d\sigma}(k)\right|^2\,dk\\
    &\ge h^{-N}\int_{\substack{k\in\Gamma \\ |k|\ge R'}}\1(|k-\xi|\le 1)\left|\hat{d\sigma}(h^{-1}k)\right|^2\,dk \\
    &\ge Ch^{-1}\int_{\substack{k\in\Gamma \\ |k|\ge R'}}\frac{\1(|k-\xi|\le 1)}{(1+|k|)^{N-1}}\,dk.
 \end{align*}
 Choosing $R=R'$, we infer that 
 $$\norm{\rho_{\cE_S\gamma_h (\cE_S)^*}}_{L^{q'}(\R^N)}\ge Ch^{-\frac{N}{q'}-1}\left(\int_{|\xi|\le R'}\left(\int_{\substack{k\in\Gamma \\ |k|\ge R'}}\frac{\1(|k-\xi|\le 1)}{(1+|k|)^{N-1}}\,dk\right)^{q'}\,d\xi\right)^{1/q'}.$$
 The double integral on the right side is easily seen to be finite. Combining this estimate with \eqref{eq:schatten-r-trial}, we obtain
 $$\frac{\norm{\rho_{\cE_S\gamma_h (\cE_S)^*}}_{L^{q'}(\R^N)}}{\norm{\gamma_h}_{\gS^r(L^2(S))}}\ge c(h^{-1})^{\frac{N}{q'}+1-\frac{N+r-1}{r}},$$
 which diverges as $h\to0$ if and only if $r>\frac{(N-1)q}{N(q-1)}$, as claimed. 
 \end{proof}

\section{Strichartz inequalities}\label{sec:strichartz}

\subsection{Laplacian case}

An important application of the restriction estimates for quadratic surfaces proved in \cite{Strichartz-77} concerns space-time decay estimates for solutions to evolution equations, which are known as \emph{Strichartz inequalities} and are a widely used tool to study nonlinear versions of these equations. For instance, when $S$ is the paraboloid
$$S=\{(\omega,\xi)\in\R\times\R^d,\,\omega=-|\xi|^2\},$$
one has the identity for all $f\in L^1(S,d\mu)$ and for all $(t,x)\in\R\times\R^d$,
$$\cE_Sf(t,x)=\frac{1}{(2\pi)^{d+1}}\int_S e^{i(t,x)\cdot(\omega,\xi)}f(\omega,\xi)d\mu(\omega,\xi)=\frac{1}{(2\pi)^{d+1}}\int_{\R^d}e^{-it|\xi|^2}e^{ix\cdot\xi}f(-|\xi|^2,\xi)\,d\xi,$$
where $d\mu$ is the measure defined at the beginning of Section \ref{sec:results-restriction}, which in the case of the paraboloid (Case I) is simply $d\mu(\omega,\xi)=(1+4|\xi|^2)^{-1/2}d\sigma(\omega,\xi)$. Hence, choosing $f(\omega,\xi)=\hat{\phi}(\xi)$ for some $\phi:\R^d\to\C$, one deduces that 
$$\cE_Sf(t,x)=\frac{1}{2\pi}\left(e^{it\Delta}\phi\right)(x),\quad\forall(t,x)\in\R\times\R^d.$$
Using that $\cE_S$ is bounded from $L^2(S,d\mu)$ to $L^{2+4/d}(\R^{d+1})$, Strichartz obtains his famous bound
\begin{equation}
 \norm{e^{it\Delta}f}_{L^{2+4/d}_{t,x}(\R\times\R^d)}\le C\norm{f}_{L^2(\R^d)},\quad\forall f\in L^2(\R^d),\quad\forall d\ge1,
\end{equation}
where $C>0$ is independent of $f$. In the same fashion, applying Theorem \ref{thm:restriction-orthonormal} with $N=d+1$ and $S$ a paraboloid (Case I), we recover Strichartz's bound for orthonormal functions \cite{FraLewLieSei-14}:

\begin{theorem}[Strichartz estimates for orthonormal functions---diagonal case]\label{thm:strichartz-orthonormal-diagonal}
 Assume that $d\ge1$. Then, for any (possibly infinite) orthonormal system $(f_j)$ in $L^2(\R^d)$ and for any $(\nu_j)\subset\C$, we have
 \begin{equation}
  \norm{\sum_j \nu_j\left|e^{it\Delta}f_j\right|^2}_{L^{1+2/d}_{t,x}(\R\times\R^d)}\le C\left(\sum_j|\nu_j|^{\frac{d+2}{d+1}}\right)^{\frac{d+1}{d+2}},
 \end{equation}
 with $C>0$ independent of $(\nu_j)$ and $(f_j)$.
\end{theorem}

Equivalently, according to Remark \ref{rk:density}, for any $\gamma\in\gS^{(d+2)/(d+1)}(L^2(\R^d))$, the inequality 
 \begin{equation}
  \norm{\rho_{e^{it\Delta}\gamma e^{-it\Delta}}}_{L^{1+2/d}_{t,x}(\R\times\R^d)}\le C\norm{\gamma}_{\gS^{(d+2)/(d+1)}(L^2(\R^d))},
 \end{equation}
 holds with $C>0$ independent of $\gamma$. As explained in the introduction, this result was proved for the first time in \cite{FraLewLieSei-14}, using a different method. We recover it as a consequence of more general restriction estimates for orthonormal functions, hence providing a different proof. Our method actually allows to go further and to answer a question left open in \cite{FraLewLieSei-14}:

\begin{theorem}[Strichartz estimates for orthonormal functions---general case]\label{thm:strichartz-orthonormal-general}
 Assume that $d\ge1$ and that $p,q\ge1$ are such that 
 $$\frac2p+\frac dq=d,\quad 1\le q<1+\frac{2}{d-1}.$$
 Then, for any (possibly infinite) orthonormal system $(f_j)$ in $L^2(\R^d)$ and for any $(\nu_j)\subset\C$, we have
 \begin{equation}
  \norm{\sum_j \nu_j\left|e^{it\Delta}f_j\right|^2}_{L^p_tL^q_x(\R\times\R^d)}\le C\left(\sum_j|\nu_j|^{\frac{2q}{q+1}}\right)^{\frac{q+1}{2q}},
 \end{equation}
 with $C>0$ independent of $(\nu_j)$ and $(f_j)$. 
\end{theorem}

Equivalently, according to Remark \ref{rk:density}, for any $\gamma\in\gS^{2q/(q+1)}(L^2(\R^d))$, the inequality 
 \begin{equation}\label{eq:strichartz-rho-general}
  \norm{\rho_{e^{it\Delta}\gamma e^{-it\Delta}}}_{L^p_tL^q_x(\R\times\R^d)}\le C\norm{\gamma}_{\gS^{2q/(q+1)}(L^2(\R^d))},
 \end{equation}
 holds with $C>0$ independent of $\gamma$. In \cite{FraLewLieSei-14}, this result was proved only for the range $1\le q\le 1+2/d$, and was shown to \emph{fail} for $q\ge1+2/(d-1)$. Hence, Theorem \ref{thm:strichartz-orthonormal-general} provides the \emph{full range} of exponents of Strichartz estimates for orthonormal functions. Notice that this range is significantly smaller than the range for a single function which is $1\le q\le 1+2/(d-2)$ for $d\ge3$ \cite{KeeTao-98}. In Section \ref{sec:hartree}, we give an application of these inequalities to the well-posedness of the non-linear Hartree equation in Schatten spaces, in the spirit of \cite{LewSab-13a}. 

Theorem \ref{thm:strichartz-orthonormal-general} follows again from a Schatten bound coupled to Lemma \ref{lemma:duality-principle}:

\begin{theorem}[Schatten bound with space-time norms]\label{thm:schatten-strichartz-mixed}
  Let $d\ge1$ and $S$ be the paraboloid
  $$S:=\{(\omega,\xi)\in\R\times\R^d,\,\omega=-|\xi|^2\}.$$
  Then, for all exponents $p,q\ge1$ satisfying the relations 
  $$\frac2p+\frac dq=1,\qquad q>d+1,$$
  we have the Schatten bound
  $$\norm{W_1T_SW_2}_{\gS^q(L^2(\R^{d+1}))}\le C\norm{W_1}_{L^p_tL^q_x(\R\times\R^d)}\norm{W_2}_{L^p_tL^q_x(\R\times\R^d)},$$
  with $C>0$ independent of $W_1,W_2$.
\end{theorem}

\begin{proof}[Proof of Theorem \ref{thm:schatten-strichartz-mixed}]
We investigate more precisely the bounds on the family $G_z$ introduced in the proof of Theorem \ref{thm:schatten-WTSW} when $S$ is the paraboloid. Strichartz \cite{Strichartz-77} uses the definition
$$G_z(\omega,\xi)=\frac{1}{\Gamma(z+1)}(\omega-|\xi|^2)^z_+,\,\forall(\omega,\xi)\in\R\times\R^d,$$
which ensures that the Fourier multiplication operator with $G_{-1}$ coincides with the operator $T_S$. As before, we have a first bound
$$\norm{T_{is}}_{L^2(\R^{d+1})\to L^2(\R^{d+1})}=\norm{G_{is}}_{L^\ii(\R^{d+1})}\le\left|\frac{1}{\Gamma(1+is)}\right|\le Ce^{\pi|s|/2}.$$
To prove that $T_{-\lambda_0+is}$ is bounded from $L^1$ to $L^\ii$, Strichartz computes explicitly the (inverse) Fourier transform of $G_z$ and obtains
$$\check{G}_z(t,x)=\pi^{-\frac{d+1}{2}}ie^{iz\pi/2}e^{-i\pi d/4}e^{-i|x|^2/4t}|t|^{-d/2}(-t-i0)^{-z-1},\,\forall (t,x)\in\R\times\R^d.$$
He deduces from this formula that $\check{G}_z$ belongs to $L^\ii_{t,x}$ when $\text{Re}\,z=-1-d/2$. We now explain how to obtain better results from this expression than the one obtained in Theorem \ref{thm:schatten-WTSW}. To do so, recall that the distribution $(-t-i0)^\lambda$ on $\R$ satisfies the identity 
$$(-t-i0)^\lambda=t_-^\lambda+e^{-i\pi\lambda}t^\lambda_+$$
for $\text{Re}\,\lambda>-1$ \cite[Ch. I, Sec. 3.6]{GelShi-64}, where $t_\pm^\lambda$ are the distributions given by the $L^1_\text{loc}$-functions 
$$t_+^\lambda=\begin{cases}
               t^\lambda &\text{for }t>0 \\
		0  &\text{for } t\le0,
              \end{cases}
  \qquad t_-^\lambda=\begin{cases}
                0 &\text{for }t\ge0 \\
		(-t)^\lambda  &\text{for } t<0.
              \end{cases}$$
In particular, the distribution $(-t-i0)^\lambda$ is also given by a $L^1_\text{loc}$-function, and we deduce the bound
$$\left|(-t-i0)^\lambda\right|\le\max\left(1,e^{\pi\text{Im}\,\lambda}\right)|t|^{\text{Re}\,\lambda},\,\quad\forall t\in\R,$$
valid for all $\text{Re}\,\lambda>-1$. In our context, we have $\lambda=-z-1$ with $z=-\lambda_0+is$, so that $\text{Re}\,\lambda=\lambda_0-1>0$. We thus deduce the bound
$$\left|\check{G}_{-\lambda_0+is}(t,x)\right|\le C\max\left(1,e^{-3\pi s/2}\right)|t|^{\lambda_0-1-d/2},\quad\forall(t,x)\in\R\times\R^d,$$
valid for all $s\in\R$ and for all $\lambda_0>1$. We now go back to the proof of Proposition \ref{prop:complex-interpolation} and provide another estimate for $\norm{W_1^{-z}T_zW_2^{-z}}_{\gS^2}$ when $z=-\lambda_0+is$ using the Hardy--Littlewood--Sobolev inequality:
\begin{align*}
 \norm{W_1^{\lambda_0-is}T_{-\lambda_0+is}W_2^{\lambda_0-is}}_{\gS^2}^2 &= \hspace{-0.2cm}\int\limits_{\R^{2(d+1)}}W_1(t,x)^{2\lambda_0}\left|\check{G}_{-\lambda_0+is}(t-t',x-x')\right|^2W_2(t',x')^{2\lambda_0}\,dxdx'dtdt'\\
  &\le C\max\left(1,e^{-3\pi s/2}\right)\int_\R\int_\R\frac{\norm{W_1(t)}_{L^{2\lambda_0}_x(\R^d)}^{2\lambda_0}\norm{W_2(t')}_{L^{2\lambda_0}_x(\R^d)}^{2\lambda_0}}{|t-t'|^{d+2-2\lambda_0}}\,dtdt'\\
 &\le C\max\left(1,e^{-3\pi s/2}\right)\norm{W_1}_{L^{\frac{4\lambda_0}{2\lambda_0-d}}_tL^{2\lambda_0}_x(\R\times\R^d)}^{2\lambda_0}\norm{W_2}_{L^{\frac{4\lambda_0}{2\lambda_0-d}}_tL^{2\lambda_0}_x(\R\times\R^d)}^{2\lambda_0},
\end{align*}
provided that $0\le d+2-2\lambda_0<1$, that is $(d+1)/2<\lambda_0\le1+d/2$. For this range of $\lambda_0$, we conclude as in the proof of Proposition \ref{prop:complex-interpolation} that
$$\norm{W_1T_{-1}W_2}_{\gS^{2\lambda_0}(L^2(\R^{d+1}))}\le C\norm{W_1}_{L^{\frac{4\lambda_0}{2\lambda_0-d}}_tL^{2\lambda_0}_x(\R\times\R^d)}\norm{W_2}_{L^{\frac{4\lambda_0}{2\lambda_0-d}}_tL^{2\lambda_0}_x(\R\times\R^d)},$$
which is the desired estimate. 
\end{proof}

\begin{remark}
 The same proof actually gives the full range of Strichartz estimates for a single function, except for the endpoints. Hence, all Strichartz estimates (and not only the diagonal ones) are implicitly contained in Strichartz's original article, except the endpoints. We are not aware that this has been observed before. 
\end{remark}

\subsection{Square root of the Laplacian case}

When $S$ is the cone $S=\{(\omega,\xi)\in\R\times\R^d,\,\omega^2=|\xi|^2\}$ endowed with the measure $d\mu(\omega,\xi)=(2|(\omega,\xi)|)^{-1}d\sigma(\omega,\xi)$, one has the identity for all $f\in L^1(S,d\mu)$ and for all $(t,x)\in\R\times\R^d$,
\begin{multline*}
  \cE_Sf(t,x)=\frac{1}{(2\pi)^{d+1}}\int_Se^{i(t,x)\cdot(\omega,\xi)}f(\omega,\xi)\,d\mu(\omega,\xi)\\
  =\frac{1}{(2\pi)^{d+1}}\int_{\R^d}e^{it|\xi|}e^{ix\cdot\xi}f(|\xi|,\xi)\frac{d\xi}{2\sqrt{2}|\xi|}+\frac{1}{(2\pi)^{d+1}}\int_{\R^d}e^{-it|\xi|}e^{ix\cdot\xi}f(-|\xi|,\xi)\frac{d\xi}{2\sqrt{2}|\xi|}.
\end{multline*}
In particular, when one chooses $f(\omega,\xi)=2\sqrt{2}|\xi|\hat{\phi}(\xi)$ if $\omega>0$ and $f(\omega,\xi)=0$ if $\omega<0$, we have the identity
$$\cE_Sf(t,x)=\frac{1}{2\pi}\left(e^{it(-\Delta)^{1/2}}\phi\right)(x),\quad\forall(t,x)\in\R\times\R^d.$$
Since $\cE_S$ is bounded from $L^2(S,d\mu)$ to $L^{2(d+1)/(d-1)}(\R^{d+1})$, we deduce the following Strichartz inequality
$$\norm{e^{it(-\Delta)^{1/2}}\phi}_{L^{2(d+1)/(d-1)}_{t,x}(\R\times\R^d)}\le C\norm{\phi}_{\dot{H}^{1/2}(\R^d)},$$
with $C>0$ independent of $\phi$. We obtain the corresponding version of this result for orthonormal functions.

\begin{theorem}[Strichartz estimates for orthonormal functions---fractional Laplacian case]\label{thm:strichartz-orthonormal-fractional}
 Assume that $d\ge1$. Then, for any (possibly infinite) orthonormal system $(f_j)$ in $\dot{H}^{1/2}(\R^d)$ and for any $(\nu_j)\subset\C$, we have
 \begin{equation}
  \norm{\sum_j \nu_j\left|e^{it(-\Delta)^{1/2}}f_j\right|^2}_{L^{\frac{d+1}{d-1}}_{t,x}(\R\times\R^d)}\le C\left(\sum_j|\nu_j|^{1+\frac{1}{d}}\right)^{\frac{d}{d+1}},
 \end{equation}
 with $C>0$ independent of $(\nu_j)$ and $(f_j)$.
\end{theorem}

We also have the operator version of this inequality
 \begin{equation}
  \norm{\rho_{e^{-it(-\Delta)^{1/2}}\gamma e^{it(-\Delta)^{1/2}}}}_{L^{\frac{d+1}{d-1}}_{t,x}(\R\times\R^d)}\le C\norm{(-\Delta)^{1/4}\gamma(-\Delta)^{1/4}}_{\gS^{1+\frac1d}(L^2(\R^d))},
 \end{equation}
 which holds with $C>0$ independent of $\gamma$.

\begin{proof}[Proof of Theorem \ref{thm:strichartz-orthonormal-fractional}]
 If $(f_j)$ is an orthonormal system in $\dot{H}^{1/2}(\R^d)$, the functions
 $$(\omega,\xi)\mapsto g_j(\omega,\xi)=\1(\omega>0)2\sqrt{2}|\xi|f_j(\xi)$$
 are orthonormal in $L^2(S,d\mu)$ as explained in the beginning of this section. We then apply Theorem \ref{thm:restriction-orthonormal} to this system, with $S$ being a cone (Case II).
\end{proof}

\subsection{Pseudo-relativistic case} Finally, when the surface $S$ is the two-sheeted hyperboloid $S=\{(\omega,\xi)\in\R\times\R^d,\,\omega^2=1+|\xi|^2\}$, with the measure $d\mu(\omega,\xi)=(2|(\omega,\xi)|)^{-1}d\sigma(\omega,\xi)$, one has the identity for all $f\in L^1(S,d\mu)$ and for all $(t,x)\in\R\times\R^d$,
\begin{align*}
  \cE_Sf(t,x) &= \frac{1}{(2\pi)^{d+1}}\int_Se^{i(t,x)\cdot(\omega,\xi)}f(\omega,\xi)\,d\mu(\omega,\xi)\\
  &=\frac{1}{(2\pi)^{d+1}}\int_{\R^d}e^{it\sqrt{1+|\xi|^2}}e^{ix\cdot\xi}f(\sqrt{1+|\xi|^2},\xi)\frac{d\xi}{2\sqrt{1+|\xi|^2}}\\
  &+\frac{1}{(2\pi)^{d+1}}\int_{\R^d}e^{-it\sqrt{1+|\xi|^2}}e^{ix\cdot\xi}f(-\sqrt{1+|\xi|^2},\xi)\frac{d\xi}{2\sqrt{1+|\xi|^2}}.
\end{align*}
In particular, when one chooses $f(\omega,\xi)=2\1(\omega>0)\sqrt{1+|\xi|^2}\hat{\phi}(\xi)$, we have the identity
$$\cE_Sf(t,x)=\frac{1}{2\pi}\left(e^{it\sqrt{1-\Delta}}\phi\right)(x),\quad\forall(t,x)\in\R\times\R^d.$$
Since $\cE_S$ is bounded from $L^2(S,d\mu)$ to $L^q(\R^{d+1})$, with $2+4/d\le q\le 2+4/(d-1)$ ($d\ge2$) and $6\le q<\ii$ ($d=1$), we deduce the following Strichartz inequality
$$\norm{e^{it\sqrt{1-\Delta}}\phi}_{L^q_{t,x}(\R\times\R^d)}\le C\norm{\phi}_{H^{1/2}(\R^d)},$$
with $C>0$ independent of $\phi$. We obtain the corresponding version of this result for orthonormal functions.

\begin{theorem}[Strichartz estimates for orthonormal functions---pseudo-relativistic case]\label{thm:strichartz-orthonormal-pseudo-relativistic}
 Assume that $d\ge1$. Let $1+2/d\le q\le 1+2/(d-1)$ if $d\ge2$ and $3\le q<\ii$ if $d=1$. Then, for any (possibly infinite) orthonormal system $(f_j)$ in $H^{1/2}(\R^d)$, and for any $(\nu_j)\subset\C$, we have
 \begin{equation}
  \norm{\sum_j \nu_j\left|e^{it\sqrt{1-\Delta}}f_j\right|^2}_{L^q_{t,x}(\R\times\R^d)}\le C\left(\sum_j|\nu_j|^{\frac{2q}{q+1}}\right)^{\frac{q+1}{2q}},
 \end{equation}
 with $C>0$ independent of $(\nu_j)$ and $(f_j)$.
\end{theorem}

We also have the operator version of this inequality
 \begin{equation}
  \norm{\rho_{e^{-it\sqrt{1-\Delta}}\gamma e^{it\sqrt{1-\Delta}}}}_{L^q_{t,x}(\R\times\R^d)}\le C\norm{(1-\Delta)^{1/4}\gamma(1-\Delta)^{1/4}}_{\gS^{\frac{2q}{q+1}}(L^2(\R^d))},
 \end{equation}
 which holds with $C>0$ independent of $\gamma$.

\begin{proof}[Proof of Theorem \ref{thm:strichartz-orthonormal-pseudo-relativistic}]
 If $(f_j)$ is an orthonormal system in $H^{1/2}(\R^d)$, the functions
 $$(\omega,\xi)\mapsto g_j(\omega,\xi)=\1(\omega>0)2(1+|\xi|^2)^{1/2}f_j(\xi)$$
 are orthonormal in $L^2(S,d\mu)$ as explained in the beginning of this section. We then apply Theorem \ref{thm:restriction-orthonormal} to this system, with $S$ being a two-sheeted hyperboloid (Case III(ii-iii)).
\end{proof}

\section{Uniform Sobolev estimates and the Limiting Absorption Principle}\label{sec:sobolev}

Our last result concerns Schatten class properties of the resolvent $(-\Delta-z)^{-1}$ of the Laplace operator on $\R^N$.

\begin{theorem}[Uniform resolvent bounds in Schatten spaces]\label{thm:uniform-sobolev-schatten}
Let $N\ge2$ and assume that
 $$\begin{cases}
  \frac43\le q\le \frac32 & \text{if } N=2,\\
  \frac{N}{2}\le q\le \frac{N+1}{2} & \text{if } N\ge3.
 \end{cases}$$
 Then, for all $z\in\C\setminus[0,\ii)$, we have the estimate
 \begin{equation}\label{eq:kenig-ruiz-sogge-schatten}
  \norm{W_1(-\Delta-z)^{-1}W_2}_{\gS^{(N-1)q/(N-q)}(L^2(\R^N))}\le C|z|^{-1+\frac{N}{2q}}\norm{W_1}_{L^{2q}(\R^N)}\norm{W_2}_{L^{2q}(\R^N)},
 \end{equation}
 where $C>0$ is independent of $W_1$, $W_2$ and $z$. 
\end{theorem}

 The estimate \eqref{eq:kenig-ruiz-sogge-schatten} with the Schatten norm replaced by the operator norm was proved by Kenig, Ruiz, and Sogge \cite{KenRuiSog-87}. Their result is only stated for $|z|\ge1$ and for $N\geq 3$ in \cite{KenRuiSog-87}. The boundedness for all $z\in\C\setminus[0,\infty)$ with the same dependence on $z$ as on the right side of \eqref{eq:kenig-ruiz-sogge-schatten} follows easily by scaling. The case $N=2$ can be treated along the same lines \cite{Frank-11}.

Similary bounds, but with pointwise assumptions $|W_j(x)| \leq C (1+|x|^2)^{-\alpha/2}$ instead of integral assumptions, can be found for instance in \cite[Prop.~7.1.22]{Yafaev-10}. We recover these results for $N/(N+1)<\alpha\leq 1$ (in dimensions $N\geq 3$, for simplicity).

\begin{remark}
 We do not know whether the restriction $q\ge 4/3$ (instead of $q>1$) in $N=2$ is technical or not; see Proposition \ref{resbound2d} below for some results in the case $1<q<4/3$.
\end{remark}

\begin{proof}
By scaling it suffices to prove the bound only for $z\neq 1$ with $|z|=1$. For such $z$ we shall prove the bounds
\begin{equation}
\label{eq:resboundint1}
\left\| W_1 (-\Delta-z)^{it} W_2 \right\|_{L^2(\R^N)\to L^2(\R^N)} \le \|W_1\|_{L^\ii(\R^N)} \|W_2\|_{L^\ii(\R^N)},\, \forall t\in\R,
\end{equation}
and
\begin{equation}
\label{eq:resboundint2}
\left\| W_1 (-\Delta-z)^{-a+it} W_2 \right\|_{\gS^2(L^2(\R^N))} \leq M_{N,a} e^{C_{N,a}'t^2} \|W_1\|_{L^\frac{4N}{N-1+2a}(\R^N)} \|W_2\|_{L^\frac{4N}{N-1+2a}(\R^N)},\,\forall t\in\R,
\end{equation}
where $a$ is an arbitrary parameter satisfying $1\leq a \leq 3/2$ if $N=2$ and $(N-1)/2\leq a\leq (N+1)/2$ if $N\geq 3$. Complex interpolation, similarly as in the proof of Proposition \ref{prop:complex-interpolation}, for the family $W_1^{\zeta} (-\Delta-z)^{-\zeta} W_2^{\zeta}$ then implies
$$
\left\| W_1 (-\Delta-z)^{-1} W_2 \right\|_{\gS^{2a}(L^2(\R^N))} \leq C_{N,a} \|W_1\|_{L^\frac{4aN}{N-1+2a}(\R^N)} \|W_2\|_{L^\frac{4aN}{N-1+2a}(\R^N)} \,.
$$
This is the claimed inequality up to the change of variables $a = q(N-1)/(2(N-q))$. If $N=2,3$ and $a=1$, then \eqref{eq:resboundint2} is already the desired bound and no complex interpolation is necessary. (We note that the bound in \eqref{eq:resboundint2} grows superexponentially, but still sub-double-exponentially with $t$. This growth can be dealt with, for example, as in \cite{Stein-56}.)

Inequality \eqref{eq:resboundint1} is straightforward by Plancherel's theorem.

For the proof of \eqref{eq:resboundint2} we note the pointwise bound, uniformly in $|z|=1$, $z\neq 1$,
\begin{equation}
\label{eq:resboundptw}
\left| (-\Delta-z)^{-a+it}(x,x') \right| \le M_{N,a} e^{C_{N,a}'t^2} |x-x'|^{-\frac{N+1}{2}+a}
\end{equation}
for $(N-1)/2 \leq a\leq (N+1)/2$ and $N\geq 2$. This bound is essentially contained in \cite{KenRuiSog-87} (see equations (2.23) and (2.25) there), and was also used in \cite{ChaSaw-90} and \cite{Frank-11}. Notice that, in \cite{KenRuiSog-87}, only the case $N\ge3$ is considered. The bound in the case $N=2$ follows from the same methods, since the dimension only enters throught the order of the modified Bessel functions $K_{N/2-a+it}$. Inequality \eqref{eq:resboundint2} then follows immediately from \eqref{eq:resboundptw} and the Hardy--Littlewood--Sobolev inequality.
\end{proof}

\begin{proposition}\label{resbound2d}
Assume that $N=2$ and that $1<q\le 4/3$. Then there is a constant $C_{q,2}'$ such that for all $z\in\C\setminus[0,\infty)$
$$
\left\| W_1 (-\Delta-z)^{-1} W_2 \right\|_{\gS^2(L^2(\R^2))} \le C_{q,2}' |z|^{-1+1/q} \|W_1\|_{L^{2q}(\R^2)} \|W_2\|_{L^{2q}(\R^2)} \,.
$$
\end{proposition}

\begin{proof}
We use the bound
$$
\left| (-\Delta-z)^{-1}(x,x') \right| \leq C_\beta |x-x'|^{-\beta}
$$
for any $0<\beta\leq 1/2$ and all $|z|=1$, $z\neq 1$; see \cite{Frank-11}, where it is pointed out that this bounds follows from those of Kenig--Ruiz--Sogge \cite{KenRuiSog-87}. The assertion now follows from the Hardy--Littlewood--Sobolev inequality.
\end{proof}

\begin{remark}\label{kerusoopt}
 The Schatten space $\gS^{(N-1)q/(N-q)}$ is optimal in \eqref{eq:kenig-ruiz-sogge-schatten}. Indeed, it is well-known that 
 $$W_1\left(\frac{1}{-\Delta-1-it}-\frac{1}{-\Delta-1+it}\right)W_2\xrightarrow[t\to0_+]{}2\pi iW_1T_{\Sph^{N-1}}W_2$$
 weakly in the sense of operators on $L^2(\R^N)$, where we recall that $T_{\Sph^{N-1}}$ was defined in Section \ref{sec:restriction}. In particular, the bound \eqref{eq:kenig-ruiz-sogge-schatten} implies a Schatten bound on $W_1T_{\Sph^{N-1}}W_2$ by the non-commutative Fatou lemma \cite[Thm. 2.7d)]{Simon-79}, for which we know the optimal exponent by Theorem \ref{thm:optimality-schatten-exponent}. A similar argument (based on Knapp's example on the sphere) shows that not even the operator norm of $W_1(-\Delta-z)^{-1} W_2$ can be bounded in terms of $\|W_1\|_{L^{2q}(\R^N)} \|W_2\|_{L^{2q}(\R^N)}$ if $2q>N+1$.
\end{remark}

\begin{remark}\label{keruso1d}
There is also a uniform resolvent bound in dimension $N=1$, namely,
 \begin{equation}\label{eq:kenig-ruiz-sogge-schatten1d}
  \norm{W_1(-\Delta-z)^{-1}W_2}_{\gS^{2}(L^2(\R))}\le \frac 12\,|z|^{-\frac{1}{2}}\norm{W_1}_{L^{2}(\R)}\norm{W_2}_{L^{2}(\R)},
 \end{equation}
for all $z\in\C\setminus[0,\infty)$. Indeed, this follows immediately from the explicit expression of the integral kernel of $(-\Delta-z)^{-1}$ in $N=1$,
 $$(-\Delta-z)^{-1}(x,y)=-\frac{e^{i\sqrt{z}|x-y|}}{2i\sqrt{z}},\quad\forall(x,y)\in\R\times\R,$$
 for all $z\in\C\setminus[0,\ii)$, where the square root is chosen such that ${\rm Im}\,\sqrt{z}>0$. Optimality of the Hilbert--Schmidt space $\gS^2$ can be proved similarly as in Remark \ref{kerusoopt}.
\end{remark}

As an application of Theorem \ref{thm:uniform-sobolev-schatten}, we prove a Limiting Absorption Principle in Schatten spaces.

\begin{theorem}[Limiting Absorption Principle in Schatten spaces]\label{thm:LAP}
Let $N\ge 2$ and assume that $V\in L^q(\R^N,\R)$ with
$$
\begin{cases}
1 < q \le 3/2 & \text{if}\ N=2 \\
\frac N 2 \le q \le \frac{N+1}{2} & \text{if}\ N\ge 3 \,.
\end{cases}
$$
Define $\alpha_q:=\max(2,(N-1)q/(N-q))$. Then $\sqrt{V} (-\Delta+V-z)^{-1} \sqrt{|V|} \in \mathfrak S^{\alpha_q}(L^2(\R^{N}))$ for every $z\in \C\setminus[0,\infty)$, where we used the notation $\sqrt{V}:=V/\sqrt{|V|}$ (with $\sqrt{V}:=0$ if $V=0$). The mapping $\C\setminus [0,\infty) \ni z\mapsto \sqrt{V} (-\Delta+V-z)^{-1} \sqrt{|V|}\in\gS^{\alpha_q}$ is analytic and extends continuously to the open interval $(0,\infty)$ (with possibly different boundary values from above and below). Moreover, under the additional assumption $q>N/2$, there is a constant $C_{N,q}$ such that if $|z|^{-1+N/{2q}} \|V\|_{L^q(\R^N)} \le C_{N,q}$ then
\begin{equation}
\label{eq:lapperturbed}
\left\| \sqrt{V} (-\Delta+V-z)^{-1} \sqrt{|V|} \right\|_{\mathfrak S^{\alpha_q}(L^2(\R^{N}))} \le 2C_{N,q} |z|^{-1+N/{2q}} \|V\|_{L^q(\R^N)} \,.
\end{equation}
If $q=N/2$ and $N\ge 3$, the bound \eqref{eq:lapperturbed} holds provided 
$|z| \ge C(V)$ for some constant $C(V)$ depending on $V$.
\end{theorem}

We expect that this theorem has applications in the context of Lieb--Thirring inequalities at positive density (see \cite{FraLewLieSei-13}) using some tools developed in \cite{FraPus-14}.

Before proving Theorem \ref{thm:LAP}, we need more information about the Birman--Schwinger operator $\sqrt{V}(-\Delta-z)^{-1}\sqrt{|V|}$.

\begin{lemma}\label{lem:extension-resolvent}
 Let $N\ge2$ and assume that $V\in L^q(\R^N,\R)$, where $q$ satisfies the assumptions of Theorem \ref{thm:LAP}. Let $\delta\subset(0,\ii)$ be a compact interval. Then, the family $A(z):=\sqrt{V}(-\Delta-z)^{-1}\sqrt{|V|}\in\gS^{\alpha_q}(L^2(\R^N))$ is analytic on the half-strips $S_\pm:=\{z\in\C,\,{\rm Re}\,z\in\mathring{\delta},\,\pm{\rm Im}\, z>0\}$. On each $S_\pm$, it is continuous up to $\bar{S_\pm}$ and we denote by $\sqrt{V}(-\Delta-\lambda\pm i0)^{-1}\sqrt{|V|}$ its extensions at $\lambda>0$. For all $z\in\bar{S_\pm}$, we have the estimate
 \begin{equation}\label{eq:BS-schatten-boundary}
    \norm{A(z)}_{\gS^{\alpha_q}}\le C|z|^{-1+\frac{N}{2q}}\norm{V}_{L^q},
 \end{equation}
 with $C$ as in \eqref{eq:kenig-ruiz-sogge-schatten} (and in particular independent of $\delta$). Finally, for all $z\in\bar{S_\pm}$, the operator $1+A(z)$ is invertible and the map $S_\pm\ni z\mapsto(1+A(z))^{-1}$ is an analytic family of bounded operators on $L^2(\R^N)$, which is continuous on $\bar{S_\pm}$. 
\end{lemma}

The proof of Lemma \ref{lem:extension-resolvent} relies on a deep theorem of Koch and Tataru \cite{KocTat-06} about the absence of embedded eigenvalues for Schr\"odinger operators.

 \begin{proof}[Proof of Lemma \ref{lem:extension-resolvent}] 
  The family $\C\setminus[0,\ii)\ni z\mapsto\sqrt{V}(-\Delta-z)^{-1}\sqrt{|V|}=A(z)\in\gS^{\alpha_q}$ is analytic: indeed, by the resolvent formula we have for any $z,z_0\in\C\setminus[0,\ii)$,
  \begin{multline*}
    \sqrt{V}(-\Delta-z)^{-1}\sqrt{|V|}-\sum_{n=0}^N(z-z_0)^n\sqrt{V}(-\Delta-z_0)^{-n-1}\sqrt{|V|}\\
    =\sqrt{V}(-\Delta-z)^{-1}(z-z_0)^{N+1}(-\Delta-z_0)^{-N-1}\sqrt{|V|}.
  \end{multline*}
  The right side of this equality goes to zero in $\gS^{\alpha_q}$ as $N\to\ii$ if $|z-z_0|$ is small enough by the Kato--Seiler--Simon inequality \cite[Thm. 4.1]{Simon-79} and the constraint $q\ge N/2$,
  \begin{multline*}
    \norm{\sqrt{V}(-\Delta-z)^{-1}(-\Delta-z_0)^{-N-1}\sqrt{|V|}}_{\gS^{\alpha_q}}\\
    \le\norm{\sqrt{|V|}(-\Delta-z_0)^{-1}}_{\gS^{2\alpha_q}}^2\norm{(-\Delta-z_0)^{-1}}^{N-1}\norm{(-\Delta-z)^{-1}}\le C^N\norm{V}_{L^q}^2.
  \end{multline*}
  The same estimate shows that the entire series we found has a nonzero radius of convergence, showing the desired analyticity. Next, let us notice that we can use the results of \cite{IonSch-06}, since the argument given in the proof of Lemma 3.5 in \cite{FraPus-14} shows that $V$ is an admissible perturbation in the sense of \cite{IonSch-06} and that $L^{2q/(q+1)}(\R^N)\subset X$, where $X$ is the Banach space defined in the introduction of \cite{IonSch-06}. Using \cite[Lemma 4.1 b)]{IonSch-06}, for each $\lambda>0$ there exists an operator $(-\Delta-\lambda\pm i0)^{-1}$ bounded from $L^{2q/(q+1)}$ to $L^{2q/(q-1)}$ such that $z\mapsto A(z)$ can be extended as a continuous family on the strips $\bar{S_\pm}$, for the weak topology of operators. Let us show that this family is actually continuous for the Schatten topology $\gS^{\alpha_q}(L^2(\R^N))$. To do so, let $z\in\bar{S_\pm}$ and $(z_n)\subset S_\pm$ such that $z_n\to z$ as $n\to\ii$. We show that $(A(z_n))$ is a Cauchy sequence for the Schatten norm, which then implies the Schatten norm continuity of $A(z)$ up to the real axis. Thus, let $\epsilon>0$. Let $W_1,\tilde{W_1}$ be bounded, compactly supported measurable functions such that we have the decompositions
  $$\sqrt{V}=W_1+W_2,\quad\sqrt{|V|}=\tilde{W_1}+\tilde{W_2},\quad\norm{W_2}_{L^{q/2}}+\norm{\tilde{W_2}}_{L^{q/2}}\le\epsilon.$$
  Using \eqref{eq:kenig-ruiz-sogge-schatten}, we may estimate
  $$\norm{A(z_n)-A(z_m)}_{\gS^{\alpha_q}}\le\norm{W_1((-\Delta-z_n)^{-1}-(-\Delta-z_m)^{-1})\tilde{W_1}}_{\gS^{\alpha_q}}+C\epsilon.$$
  By \cite[Prop. VII.1.22]{Yafaev-10}, the family $z\mapsto W_1(-\Delta-z)^{-1}\tilde{W_1}$ is continuous on $\bar{S_\pm}$ for the $\gS^{\alpha_q}$-topology, and hence for $n,m$ large enough, we have 
  $$\norm{W_1((-\Delta-z_n)^{-1}-(-\Delta-z_m)^{-1})\tilde{W_1}}_{\gS^{\alpha_q}}\le\epsilon.$$
  This shows that $(A(z_n))$ is a Cauchy sequence for the $\gS^{\alpha_q}$-topology, and hence $z\mapsto A(z)\in\gS^{\alpha_q}$ is continuous up to the boundary. We conclude, in particular, that $\sqrt{V}(-\Delta-\lambda\pm i0)^{-1}\sqrt{|V|}$ belongs to $\gS^{\alpha_q}$ for all $\lambda>0$. This also shows that the estimate \eqref{eq:BS-schatten-boundary} carries over to the real axis. We then apply analytic Fredholm theory \cite[Thm. I.4.2 \& I.4.3]{Yafaev-92} to the family $(A(z))$ in the strips $\bar{S_\pm}$ to infer that $z\mapsto(1+A(z))^{-1}$ is a meromorphic family of bounded operators on $S_\pm$ with poles at the points $z$ where $-1\in\sigma(A(z))$. Furthermore, this family is continuous up to the real axis, except at the points $\lambda\in\delta$ such that $-1\in\sigma(A(\lambda))$.
  
  It thus only remains to show that $-1\notin\sigma(A(z))$ for all $z\in\bar{S_\pm}$. When ${\rm Im}\,z\neq0$, this follows from a simple argument similar to the beginning of the proof of Lemma 4.6 in \cite{IonSch-06} based on the fact that $V$ is real-valued. We now consider real positive $z$. Assume that there are $\lambda>0$, a sign $\pm$, and an $f\in L^2(\R^N)$ satisfying 
  $$\sqrt{V}R_0(\lambda)\sqrt{|V|}f=-f,$$
  where we used the notation $R_0(\lambda)=(-\Delta-\lambda\pm i0)^{-1}$. We need to show that $f\equiv 0$. Let us define $g:=R_0(\lambda)\sqrt{|V|}f$. Since $f\in L^2$ and $V\in L^q$, H\"older's inequality implies that $\sqrt{|V|}f\in L^{2q/(q+1)}(\R^N)$ and therefore, by the boundedness of $R_0(\lambda)$ from $L^{2q/(q+1)}$ to $L^{2q/(q-1)}$ (this is the Kenig--Ruiz--Sogge bound mentioned after Theorem \ref{thm:uniform-sobolev-schatten}), we have $g\in L^{2q/(q-1)}(\R^N)$. Moreover, again by H\"older's inequality, $Vg\in L^{2q/(q+1)}(\R^N)$. We note that the equation for $f$ implies that
\begin{equation}\label{eq:eq-on-g}
    R_0(\lambda)Vg=-g.
 \end{equation}  
  Using this equation together with the integrability properties of $g$ and $Vg$, we can show that $g$ satisfies the Schrödinger equation $(-\Delta+V)g=zg$ in the sense of distributions on $\R^N$. Since $g\in L^{2q/(q-1)}(\R^N)$ and $Vg\in L^{2q/(q+1)}(\R^N)$, we deduce that $g\in W^{2,2q/(q+1)}_{\text{loc}}(\R^N)\subset H^1_\text{loc}(\R^N)$. If we can show that $g\in L^2(\R^N)$ (or $|x|^{-1/2+\epsilon} g\in L^2(\R^N)$ for some $\epsilon>0$), then \cite[Thm. 3]{KocTat-06} will imply that $g\equiv 0$. Therefore also $f= -\sqrt{V}R_0(\lambda)\sqrt{|V|}f = -\sqrt{V} g\equiv 0$, and $-1\notin\sigma(A(\lambda))$.
  
Thus, it remains to prove that $g\in L^2(\R^N)$. As mentioned before, $Vg\in L^{2q/(q+1)}(\R^N)$, and therefore $Vg\in X$, the space introduced in \cite{IonSch-06}. By \cite[Lemma 4.1 a,b)]{IonSch-06}, the resolvent $R_0(\lambda)$ is a bounded operator from $X$ to $X^*$, showing that $g=-R_0(\lambda)Vg\in X^*$. Using now \eqref{eq:eq-on-g} and \cite[Lemma 4.4]{IonSch-06}, we obtain the decay estimate
 $$\norm{(1+|x|^2)^Mg}_{X^*}<\ii$$
 for all $M\ge0$. Since $X^*\subset L^{2q/(q-1)}(\R^N)$ and writing for $M$ large enough
 $$g=(1+|x|^2)^{-M}\times(1+|x|^2)^Mg,$$
 we deduce from Hölder's inequality that $g\in L^2(\R^N)$. This completes the proof.
\end{proof}

We are now ready to prove Theorem \ref{thm:LAP}.

\begin{proof}[Proof of Theorem \ref{thm:LAP}]
  We make use of the identity
  \begin{equation}\label{eq:decomp-resolvent}
    \sqrt{V}(-\Delta+V-z)^{-1}\sqrt{|V|}=\frac{1}{1+\sqrt{V}(-\Delta-z)^{-1}\sqrt{|V|}}\sqrt{V}(-\Delta-z)^{-1}\sqrt{|V|}.
  \end{equation}
  By Lemma \ref{lem:extension-resolvent}, we know that the maps $$z\mapsto\frac{1}{1+\sqrt{V}(-\Delta-z)^{-1}\sqrt{|V|}}\in\B(L^2),\quad z\mapsto\sqrt{V}(-\Delta-z)^{-1}\sqrt{|V|}\in\gS^{\alpha_q}$$
  are analytic on $\C\setminus[0,\ii)$ and extend continuously to $(0,\ii)$ with possibly different boundary values from above and from below. We are thus left to prove \eqref{eq:lapperturbed}. First assume that $q>N/2$. Then it follows from Theorem \ref{thm:uniform-sobolev-schatten} that
  $$
  \left\| \sqrt{V} (-\Delta-z)^{-1} \sqrt{|V|} \right\|_{L^2(\R^N)\to L^2(\R^N)} \le \frac 12
  $$
  provided $C|z|^{-1+N/{2q}}\|V\|_{L^q(\R^N)} \le 1/2$. Thus, for such $z$,
  $$
  \left\| \left( 1+ \sqrt{V} (-\Delta-z)^{-1} \sqrt{|V|} \right)^{-1} \right\|_{L^2(\R^N)\to L^2(\R^N)} \le  2 \,.
  $$
  The claimed bound then follows from the identity \eqref{eq:decomp-resolvent} and the inequality \eqref{eq:BS-schatten-boundary}.

  Now assume that $q=N/2$ and $N\geq 3$. In this case we write $\sqrt V = W_1 + W_2$ and $\sqrt{|V|}=\tilde W_1 + \tilde W_2$ as in the proof of Lemma \ref{lem:extension-resolvent}. Then using again Theorem \ref{thm:uniform-sobolev-schatten}
  $$
  \left\| \sqrt{V} (-\Delta-z)^{-1} \sqrt{|V|} \right\|_{L^2(\R^N)\to L^2(\R^N)} \le \left\| W_1 (-\Delta-z)^{-1} \tilde W_1 \right\|_{L^2(\R^N)\to L^2(\R^N)} + C\epsilon
  $$
  for all $z\in\C\setminus[0,\infty)$. Since
  $$
  \left\| W_1 (-\Delta-z)^{-1} \tilde W_1 \right\|_{L^2(\R^N)\to L^2(\R^N)} \to 0
  $$
  as $|z|\to\infty$, as we have just shown (note that $W_1,\tilde W_1 \in L^q(\R^N)$ with $q>N/2$), we can argue as before and obtain the result in this case as well. This finishes the proof of Theorem \ref{thm:LAP}.
 \end{proof}

\section{Application of Strichartz estimates: global well-posedness for the Hartree equation in Schatten spaces}\label{sec:hartree}

We illustrate the usefulness of the Strichartz estimates obtained in Theorem \ref{thm:strichartz-orthonormal-general} by showing well-posedness results in Schatten spaces for the non-linear Hartree equation. The main point here is that we can consider a system of \emph{infinitely many} equations and that we do not even need a trace class assumption. This is of crucial importance when studying the dynamics of quantum gases at positive density \cite{LewSab-13a,LewSab-13b}. Using our improved set of Strichartz exponents from Theorem \ref{thm:strichartz-orthonormal-general} we can extend the previous work of one of us \cite[Chap. 4, App. A]{Sabin-thesis}. Our exposition here is somehow sketchy. Details as well as similar consequences of Strichartz inequalities for the fractional Laplacian and for the pseudo-relativistic operator will be addressed in a forthcoming work of the second author.

\begin{theorem}\label{thm:gwp-hartree}
 Let $d\ge1$, $1\le q<1+2/(d-1)$, $p\ge1$ such that $2/p+d/q=d$ and $w\in L^{q'}_x(\R^d)$. Then, for any $\gamma_0\in\gS^{2q/(q+1)}$, there exists a unique $\gamma\in C^0_t(\R,\gS^{2q/(q+1)})$ satisfying $\rho_\gamma\in L^p_{{\rm loc},t}(\R,L^q_x(\R^d))$ and
 $$\begin{cases}
    i\partial_t\gamma = [-\Delta+w*\rho_\gamma,\gamma],\\
    \gamma_{|t=0} = \gamma_0.
   \end{cases}
 $$
\end{theorem}

This result is a consequence of homogeneous and inhomogeneous Strichartz estimates. The homogeneous part is the content of Theorem \ref{thm:strichartz-orthonormal-general}, while the inhomogeneous one can be deduced from Theorem \ref{thm:schatten-strichartz-mixed} using the same method as the proof of Corollary 1 in \cite{FraLewLieSei-14}:

\begin{theorem}\label{thm:strichartz}
 Let $d\ge1$, $1\le q<1+2/(d-1)$, $p\ge1$ such that $2/p+d/q=d$, and $\gamma_0\in\gS^{2q/(q+1)}$. Let $\gamma=\gamma(t)$ be the solution to the equation
 $$\begin{cases}
    i\partial_t\gamma = [-\Delta,\gamma] + R(t),\\
    \gamma_{|t=0} = \gamma_0.
   \end{cases}$$
 Then, the inequality
 $$\norm{\rho_{\gamma(t)}}_{L^p_t(\R,L^q_x(\R^d))}\le C_{{\rm Stri}}\left(\norm{\gamma_0}_{\gS^{2q/(q+1)}}+\norm{\int_{\R}e^{-is\Delta}|R(s)|e^{is\Delta}\,ds}_{\gS^{2q/(q+1)}}\right)$$
 holds for some constant $C_{{\rm Stri}}>0$ independent of $\gamma_0$ and $R$. 
\end{theorem}

\begin{proof}[Proof of Theorem \ref{thm:gwp-hartree}]
 We use a standard fixed point method on the Duhamel formulation of the Hartree equation, see \cite{LewSab-13a} for details. Let $R>0$ such that $\norm{\gamma_0}_{\gS^{2q/(q+1)}}\le R$, and let $T=T(R)>0$ to be chosen later on. We define a map 
 $$\Phi(\gamma,\rho)=\left(\Phi_1(\gamma,\rho),\rho[\Phi_1(\gamma,\rho)]\right),$$
 and we show that, for a suitable $T$, $\Phi$ is a contraction on the space 
 \begin{multline*}
    X:=\left\{(\gamma,\rho)\in C^0_t([0,T],\gS^{2q/(q+1)})\times L^p_t([0,T],L^q_x(\R^d)),\right.\\
    \left.\norm{\gamma}_{C^0\gS^{2q/(q+1)}}+\norm{\rho}_{L^p_tL^q_x}\le 4\max(1,C_{\text{Stri}})R\right\}.
 \end{multline*}
 The map $\Phi_1$ is defined as 
 $$\Phi_1(\gamma,\rho)(t)=e^{it\Delta}\gamma_0e^{-it\Delta}-i\int_0^te^{i(t-s)\Delta}[w*\rho(s),\gamma(s)]e^{i(s-t)\Delta}\,ds.$$
 For all $(\gamma,\rho)\in X$, we have
 \begin{align*}
    \norm{\Phi_1(\gamma,\rho)}_{C^0_t\gS^{2q/(q+1)}} &\le R + 2\int_0^T\norm{w*\rho(s)}_{L^\ii_x}\norm{\gamma(s)}_{\gS^{2q/(q+1)}}\,ds\\
    &\le R + 2T^{1/p'}\norm{w}_{L^{q'}_x}\norm{\rho}_{L^p_tL^q_x}\norm{\gamma}_{C^0_t\gS^{2q/(q+1)}}\\
    &\le R + 8T^{1/p'}\norm{w}_{L^{q'}_x}\max(1,C_{\text{Stri}}^2)R^2.
 \end{align*}
 By Theorem \ref{thm:strichartz}, we also have 
 $$
   \norm{\rho[\Phi_1(\gamma,\rho)]}_{L^p_tL^q_x}\le C_\text{Stri}R+8C_\text{Stri}T^{1/p'}\norm{w}_{L^{q'}_x}\max(1,C_{\text{Stri}}^2)R^2.
 $$
Hence, for $T=T(R)>0$ small enough, $\Phi$ maps $X$ to itself. A similar argument shows that $\Phi$ is a contraction on $X$, and thus has a unique fixed point on $X$ which is a solution to the Hartree equation on $[0,T]$. We may extend it to a maximal solution on some interval $[0,T_{\text{max}})$, and we have the blow-up criterion given by the local theory
$$T_{\text{max}}<\ii\Longrightarrow \norm{\gamma(t)}_{\gS^{2q/(q+1)}}\xrightarrow[t\to T_{\text{max}}]{}+\ii.$$
By the result of Yajima \cite{Yajima-87}, knowing that the potential $w*\rho_\gamma$ belongs to  $L^p_{\text{loc},t}L^\ii_x$ implies that there exists a unitary operator $U(t)$ on $L^2_x(\R^d)$ such that $\gamma(t)=U(t)\gamma_0U(t)^*$ for all $t$. In particular, the $\gS^{2q/(q+1)}$ norm of $\gamma(t)$ is a conserved quantity and cannot blow-up, thus leading to global solutions.
\end{proof}

\section{Application of uniform Sobolev estimates: eigenvalues of Schrödinger operators with complex-valued potentials}\label{sec:LT}

As an application of the uniform Sobolev estimates obtained in Theorem \ref{thm:uniform-sobolev-schatten}, we prove Lieb--Thirring-type inequalities for the discrete spectrum of a Schrödinger operator $-\Delta+V$ where $V$ is a complex potential belonging to $L^q(\R^N,\C)$ and $q$ satisfies the assumptions of Theorem \ref{thm:LAP}. It was noticed in \cite{Frank-11} that the uniform Sobolev estimates of Kenig, Ruiz, and Sogge could be used to control the size of eigenvalues of non self-adjoint Schrödinger operators. More precisely, if $\lambda\in\C\setminus[0,\ii)$ is an eigenvalue of $-\Delta+V$, then
\begin{equation}\label{eq:complex-frank}
  |\lambda|^\gamma\le D_{\gamma,N}\int_{\R^N}|V(x)|^{\gamma+N/2}\,dx,
\end{equation}
for some constant $D_{\gamma,N}>0$, for all $0<\gamma\le1/2$ and for all $N\ge2$; see \cite{AbrAslDav-01} for the earlier result for $\gamma=1/2$ and $N=1$. Since the Kenig--Ruiz--Sogge bound implies a control on the size of a single eigenvalue, it is natural to expect that the Schatten bounds of Theorem \ref{thm:uniform-sobolev-schatten} would provide bounds on sums of such eigenvalues. This is the content of the following result.

\begin{theorem}\label{thm:LT-complex}
  Let $N\ge 1$ and assume that $V\in L^q(\R^N,\C)$ with 
 $$\begin{cases}
 q=1 & \text{if } N=1,\\
  1<q\le 3/2 & \text{if } N=2,\\
  \frac{N}{2}\le q\le \frac{N+1}{2} & \text{if } N\ge3.
 \end{cases}$$ 
 Denote by $Z$ the (discrete) set of eigenvalues of $-\Delta+V$ in $\C\setminus[0,\ii)$, and for $\lambda\in Z$ let $m_\lambda$ be the corresponding \emph{algebraic} multiplicity. Then, we have the following bounds:
 \begin{itemize}
  \item If $N/2<q\le (N+1)/2$, then
   \begin{equation}\label{eq:LT-1}
      \sum_{\lambda\in Z}m_\lambda \,\frac{d(\lambda,[0,\ii))}{|\lambda|^{(1-\epsilon)/2}} \le A_{N,q,\epsilon}\, \|V\|_{L^q(\R^N,\C)}^{(1+\epsilon)q/(2q-N)} \,,
   \end{equation}
   where 
   $$\begin{cases}
   \epsilon>1 & \text{if}\ N=1,\\
   \epsilon\geq 0 & \text{if}\ N\ge2\ \text{and}\ N/2<q<N^2/(2N-1),\\
   \epsilon>\frac{(2N-1)q-N^2}{N-q} & \text{if}\ N\ge2\ \text{and}\ N^2/(2N-1)\le q\le (N+1)/2.
  \end{cases}$$
  \item If $q=N/2$ and $N\geq 3$, then
  \begin{equation}\label{eq:LT-3}
    \sum_{\lambda\in Z}m_\lambda\,\frac{{\rm Im}\,\sqrt{\lambda}}{1+|\lambda|}<\ii \,,
  \end{equation}
  where the branch of the square root is chosen to have positive imaginary part. 
 \end{itemize}
Here, $d(\lambda,[0,\ii))$ denotes the distance of $\lambda$ to $[0,\ii)$. The constants $A_{N,q,\epsilon}$ are independent of $V$.
\end{theorem}

Eigenvalues of Schr\"odinger operators with complex-valued potentials may in principle accumulate at any point in $[0,\ii)$ and at infinity, and \eqref{eq:LT-1} and \eqref{eq:LT-3} give quantitative information on the accumulation rate. (Note that \eqref{eq:complex-frank} excludes accumulation at infinity if $q>N/2$; for $q=N/2$, the bound \eqref{eq:LT-3} controls such a possible divergence.) Concerning accumulation at a point in $(0,\infty)$, our bounds are better than the bounds of \cite{LapSaf-09,DemHanKat-09}. Indeed, \eqref{eq:LT-1} shows that the sequence of imaginary parts of eigenvalues accumulating at a point in $(0,\infty)$ is in $\ell^1$, while \eqref{eq:LT-3} even shows that it is in $\ell^{1/2}$. In \cite{LapSaf-09,DemHanKat-09}, the best result gives that such a sequence is only in $\ell^p$ for some larger exponent $p$. Concerning accumulation of eigenvalues at zero, sometimes \eqref{eq:LT-1} and sometimes \cite{DemHanKat-09} give better results, depending on $q$ and $N$. We emphasize, however, that \cite{LapSaf-09,DemHanKat-09} require a lower bound on the real part of $V$ or a bound on the numerical range of the operator. In Theorem \ref{thm:LT-complex} we are able to remove those conditions for $N/2\leq q\le (N+1)/2$.

Note that, in view of \eqref{eq:complex-frank}, inequality \eqref{eq:LT-1} is stronger the smaller $\epsilon$ is.

The bounds from Theorem \ref{thm:LT-complex} should be compared with the usual Lieb--Thirring inequality for eigenvalues of self-adjoint Schrödinger operators, which states that when $V\in L^q(\R^N,\R)$, then
\begin{equation}\label{eq:usual-LT}
 \sum_{\lambda\in Z}m_\lambda|\lambda|^{\gamma}=\sum_{\lambda\in Z}m_\lambda d(\lambda,[0,\ii))^{\gamma}\le K_{\gamma,N}\int_{\R^N}|V(x)|^{\gamma+N/2}\,dx
\end{equation}
for all $\gamma\ge0$ when $N\ge3$, $\gamma>0$ when $N=2$, $\gamma\ge1/2$ when $N=1$. Our assumptions on $V$ corresponds to the range $0\le\gamma\le1/2$ in the previous inequality and, in fact, one can rewrite \eqref{eq:LT-1} as
\begin{equation}\label{eq:LT-1b}
 \left( \sum_{\lambda\in Z}m_\lambda \frac{d(\lambda,[0,\ii))}{|\lambda|^{(1-\epsilon)/2}}   \right)^{2\gamma/(1+\epsilon)}
 \le A_{N,\gamma+N/2,\epsilon}^{2\gamma/(1+\epsilon)} \, \int_{\R^N} |V(x)|^{\gamma+N/2} \,dx \,.
\end{equation}
Since $\left( \sum_{\lambda\in Z}m_\lambda \frac{d(\lambda,[0,\ii))}{|\lambda|^{(1-\epsilon)/2}}   \right)^{2\gamma/(1+\epsilon)} = \left( \sum_{\lambda\in Z}m_\lambda |\lambda|^{(1+\epsilon)/2} \right)^{2\gamma/(1+\epsilon)}\leq \sum_{\lambda\in Z}m_\lambda|\lambda|^{\gamma}$ if $Z\subset(-\ii,0)$ and $0<\gamma\leq 1/2$, \eqref{eq:LT-1b} is weaker than \eqref{eq:usual-LT} in the self-adjoint case.

Eigenvalues bounds for Schrödinger operators with complex-valued potentials have been the topic of many works, see for instance \cite{AbrAslDav-01,FraLapLieSei-06,LapSaf-09,DemHanKat-09,Frank-11} and references therein. To prove Theorem \ref{thm:LT-complex}, we use a method developed by Demuth, Hansmann and Katriel relying on estimates on zeros of holomorphic functions and nicely exposed in the review \cite{DemHanKat-13}. The basic idea is that any eigenvalue $z_0\in\C\setminus[0,\ii)$ of the operator $-\Delta+V$ corresponds to a ``zero'' of the analytic function $z\mapsto1+\sqrt{V}(-\Delta-z)^{-1}\sqrt{|V|}$. Hence, estimates on sums of eigenvalues of $-\Delta+V$ amount to estimates on sums of zeros of holomorphic functions. The first result of this kind is Jensen's inequality \cite[Sec. II.2]{Garnett-81}, which states that the zeros $(z_n)$ of a bounded analytic function on the unit disk satisfy the bound
$$\sum_n(1-|z_n|)<\ii,$$
which may be compared to be bounds obtained in Theorem \ref{thm:LT-complex}. As explained in \cite{DemHanKat-13}, one may obtain similar results for analytic functions that blow up at the boundary of the unit disk, according to the rate of this blow-up. This is the content of a theorem of Borichev, Golinskii, and Kupin \cite{BorGolKup-09} which generalizes Jensen's inequality to functions that may blow up at the boundary; see also \cite{FavGol-09} for a more general bound and an alternative proof. We use these versions of Jensen's inequality to prove Theorem \ref{thm:LT-complex}, in the spirit of \cite{DemHanKat-13}. The reason why we are able to significantly improve upon the results from \cite{DemHanKat-13} is that we have the uniform resolvent bounds in Schatten spaces from Theorem \ref{thm:uniform-sobolev-schatten}.

Finally, let us mention that we are not able to obtain an explicit bound in terms of $\norm{V}_{L^{N/2}}$ in \eqref{eq:LT-3} because we do not have any quantitative control on the size of $\sqrt{V}(-\Delta-z)^{-1}\sqrt{|V|}$ 
when $z$ is large, as reflected by the bound \eqref{eq:kenig-ruiz-sogge-schatten} which is uniform in $z$ when $q=N/2$. We will see in the following proof that it is a crucial input in the case $q>N/2$, which is lacking when $q=N/2$.

\begin{proof}[Proof of Theorem \ref{thm:LT-complex}]
 As we saw in Section \ref{sec:sobolev}, the map
 $$\C\setminus[0,\ii)\ni z\mapsto 1+\sqrt{V}(-\Delta-z)^{-1}\sqrt{|V|}\in1+\gS^{\alpha_q}$$
 is analytic for $N\geq 2$. The same is true, with an even easier proof, if $N=1$ ($\alpha_q=2$ in this case). Hence, for any $N\geq 1$, the map
 $$\C\setminus[0,\ii)\ni z\mapsto h(z):=\Det_{\lceil \alpha_q\rceil}\left(1+\sqrt{V}(-\Delta-z)^{-1}\sqrt{|V|}\right)\in\C$$
 is also analytic, where $\Det_{\lceil \alpha_q\rceil}$ denotes the regularized determinant, see for instance \cite[Sec. 2.2]{DemHanKat-09}, and $\lceil \alpha_q\rceil$ denotes the smallest integer which is less than or equal to $\alpha_q$. Furthermore, we have the inequality
 \begin{equation}\label{eq:reg-determinant}
    \log\left|\Det_{\lceil \alpha_q\rceil}\left(1+\sqrt{V}(-\Delta-z)^{-1}\sqrt{|V|}\right)\right|\le C\norm{\sqrt{V}(-\Delta-z)^{-1}\sqrt{|V|}}_{\gS^{\alpha_q}}^{\alpha_q} \,;
 \end{equation}
see, e.g., \cite[Thm 9.2]{Simon-79}, where this is proved with $\alpha_q$ replaced by $\lceil \alpha_q\rceil$ on the right side. A simple modification of that proof, however, yields \eqref{eq:reg-determinant} as stated; see also \cite[Lem. XI.9.22]{DunSch}.

As explained for instance in \cite[Thm. 21]{LatSuk-10}, the zeros of $h$ (counted with multiplicity) are exactly the eigenvalues of $-\Delta+V$ (counted with \emph{algebraic} multiplicity). Hence, we will prove bounds on sums of zeros of $h$, and to do so we use Jensen-type inequalities.
 
We begin with the case $q=N/2$, $N\ge3$, for which we can apply directly the usual Jensen inequality. In this case, by \eqref{eq:reg-determinant} and Theorem \ref{thm:uniform-sobolev-schatten} we can bound 
 $$\log|h(z)|\le C \|V\|_{L^{N/2}(\R^N,\C)}^N$$
 for all $z\in\C\setminus[0,\ii)$, with $C>0$ independent of $z$ and $V$. As a consequence, we can apply Jensen's inequality to the map $\mathbb{H}\ni w\mapsto h(w^2)$, where $\mathbb{H}$ denotes the upper half space. The version of Jensen's inequality for analytic functions on the upper half space bounded up to the real axis can be found for instance in \cite[Eq. (2.2.3)]{Garnett-81}, which gives exactly \eqref{eq:LT-3}.
 
When $N/2<q\le(N+1)/2$, the Schatten bounds of Theorem \ref{thm:uniform-sobolev-schatten}, Proposition \ref{resbound2d} (for $N\geq 2$) and Remark \ref{keruso1d} (for $N=1$) depend on $z$: according to \eqref{eq:reg-determinant} they yield
 \begin{equation}\label{eq:bound-log-h}
    \log|h(z)|\le C |z|^{(N-2q)\alpha_q/(2q)}\, \|V\|_{L^q(\R^N,\C)}^{\alpha_q},
 \end{equation}
where we recall that
$$\alpha_q=\left\{\begin{array}{ll}
                   2 & \text{if}\ N=1,\\
		   \max\left(2,\frac{(N-1)q}{N-q}\right) & \text{if}\ N\ge2.
                  \end{array}\right.
$$
In order to take this $z$-dependence into account we shall replace the ordinary Jensen inequality by a version of \cite{BorGolKup-09} for functions that blow-up at some point of the boundary. A special case of this result states that if $g$ is an analytic function on the unit disk $\mathbb D$ with $|g(0)|=1$ and $\log|g(w)| \leq D |1+w|^{-\alpha}$, then
 \begin{equation}
 \label{eq:borgolkup}
 \sum_{w\in Z_g}m_w(1-|w|) |1+w|^{(\alpha-1+\delta)_+} \le M_{\alpha,\delta} D
 \end{equation}
 for all $\delta>0$, $\alpha\ge 0$ and a constant $M_{\alpha,\delta}$ depending only on $\alpha$ and $\delta$. Here $Z(g)$ denotes the set of zeros of $g$ and $m_w$ denotes the multiplicity of the zero $w$. Since this result is stated for analytic functions on $\mathbb D$ and $h$ is defined on $\C\setminus[0,\ii)$, we use a conformal map as in \cite[Sec. 4.4]{DemHanKat-09}: let
 $$\psi:\mathbb{D}\ni w\mapsto a\left(\frac{1+w}{1-w}\right)^2\in\C\setminus[0,\ii),$$
 where $a<0$ is chosen such that $h(a)\neq0$. More precisely, by \cite[Thm. 9.2]{Simon-79}) (with the same modification as discussed above),
 \begin{align*}
 \label{eq:reg-determinant}
 & \left|\Det_{\lceil \alpha_q\rceil}\left(1+\sqrt{V}(-\Delta-z)^{-1}\sqrt{|V|}\right) - 1 \right| \\
 & \le \norm{\sqrt{V}(-\Delta-z)^{-1}\sqrt{|V|}}_{\gS^{\alpha_q}}^{\alpha_q} \exp\left( C_q \left( \norm{\sqrt{V}(-\Delta-z)^{-1}\sqrt{|V|}}_{\gS^{\alpha_q}}^{\alpha_q} +1 \right)^{\alpha_q} \right) \,.
 \end{align*}
 By continuity there is an $\epsilon_q>0$ such that $x \exp( C_q(x+1)^{\alpha_q}) \leq 1/2$ for all $0\leq x\leq \epsilon_q$. Thus, the uniform resolvent bound in Schatten spaces from Theorem \ref{thm:uniform-sobolev-schatten} and Remark \ref{keruso1d} implies that
 $$
 \left| \left|h(z)\right| - 1\right| \le \left| h(z) - 1\right| \le \frac12
 $$
provided
\begin{equation}
\label{eq:aconstraint}
C_{N,q}^{\alpha_q} |z|^{(N-2q)\alpha_q/(2q)} \|V\|_{L^q}^{\alpha_q} \le \epsilon_q \,.
\end{equation}
(Here $C_{N,q}$ denotes the constant on the right side in Theorem \ref{thm:uniform-sobolev-schatten} and Remark \ref{keruso1d}.) From now on we will assume that $a<0$ and that $z=a$ satisfies \eqref{eq:aconstraint}. Thus, $\log|h(a)| \geq -\log 2$ and, in particular, $h(a)\neq 0$, as desired.

Define now the map
 $$g:\mathbb{D}\ni w\mapsto h(\psi(w))/h(a)\,.$$
This is an analytic function from the unit disk to $\C$, and by \eqref{eq:bound-log-h} and \eqref{eq:aconstraint} it satisfies the bound
 \begin{align*}
    \log|g(w)| &\le C 2^{(4q-2N)\alpha_q/(2q)} \frac{|a|^{(N-2q)\alpha_q/(2q)}\|V\|_{L^q}^{\alpha_q} }{|1+w|^{(4q-2N)\alpha_q/(2q)}}-\log|h(a)|\\
    &\le \frac{C 2^{(4q-2N)\alpha_q/(2q)} \epsilon_q C_{N,q}^{-\alpha_q} + 2^{(4q-2N)\alpha_q/(2q)} \log 2 }{|1+w|^{(4q-2N)\alpha_q/(2q)}}
 \end{align*}
 for all $w\in\mathbb{D}$. Note that this bound is independent of $V$ and $a$.
 
We apply \eqref{eq:borgolkup} with $\alpha=(4q-2N)\alpha_q/(2q)$ and $\delta=1-\alpha+\epsilon$, where 
$$\begin{cases}
   \epsilon>1 & \text{if}\ N=1,\\
   \epsilon\geq 0 & \text{if}\ N\ge2\ \text{and}\ N/2<q<N^2/(2N-1),\\
   \epsilon>\frac{(2N-1)q-N^2}{N-q} & \text{if}\ N\ge2\ \text{and}\ N^2/(2N-1)\le q\le (N+1)/2.
  \end{cases}
$$
With this choice we have $\delta>0$ in any case. We obtain
 $$\sum_{w\in Z(g)}m_w(1-|w|)|1+w|^\epsilon \le M_{\alpha,\delta}\left(C 2^{(4q-2N)\alpha_q/(2q)} \epsilon_q C_{N,q}^{-\alpha_q} + 2^{(4q-2N)\alpha_q/(2q)} \log 2\right)  =: M_{N,q,\epsilon}'  \,.$$
The corresponding result for eigenvalues of $-\Delta+V$ is thus 
 $$\sum_{\lambda\in Z}m_\lambda(1-|\psi^{-1}(\lambda)|)|1+\psi^{-1}(\lambda)|^\epsilon \le M_{N,q,\epsilon}' \,.$$
We now follow some of the arguments of \cite{DemHanKat-09}. By Koebe's distortion theorem (see, e.g., \cite[p. 9]{Pommerenke}), for any $\lambda\in\C\setminus[0,\infty)$,
$$
1-|\psi^{-1}(\lambda)| \ge \frac{d(\lambda,[0,\infty))\ |(\psi^{-1})'(\lambda)|}{2} \,.
$$
Now
$$
(\psi^{-1})'(\lambda) = \frac{-\sqrt{|a|}}{\sqrt{-\lambda}(\sqrt{-\lambda} + \sqrt{|a|})^2}
\qquad\text{and}\qquad
1+\psi^{-1}(\lambda) = \frac{2\sqrt{-\lambda}}{\sqrt{-\lambda}+\sqrt{|a|}}
$$
with the convention that $\re\sqrt{-\lambda}>0$ for $\lambda\in\C\setminus[0,\infty)$. Thus,
\begin{align*}
(1-|\psi^{-1}(\lambda)|)|1+\psi^{-1}(\lambda)|^\epsilon & \ge \frac{d(\lambda,[0,\infty))\ \sqrt{|a|}}{2^{1-\epsilon}\, |\lambda|^{(1-\epsilon)/2 }\, |\sqrt{-\lambda}+\sqrt{|a|}|^{2+\epsilon}} \\
& \ge \frac{d(\lambda,[0,\infty))\ \sqrt{|a|}}{2^{1-\epsilon}\, |\lambda|^{(1-\epsilon)/2 }\, (\sqrt{|\lambda|}+\sqrt{|a|})^{2+\epsilon}} \,.
\end{align*}
As we explained before, if $z$ satisfies \eqref{eq:aconstraint} then $h(z)\neq 0$. Thus, all eigenvalues $\lambda\in Z$ satisfy $C_{N,q}^{\alpha_q} |\lambda|^{(N-2q)\alpha_q/(2q)} \|V\|_{L^q}^{\alpha_q}>\epsilon_q$, and therefore $|\lambda|<|a|$. Using this, we obtain
\begin{align*}
\sum_{\lambda\in Z}m_\lambda\, (1-|\psi^{-1}(\lambda)|)|1+\psi^{-1}(\lambda)|^\epsilon 
\ge \sum_{\lambda\in Z} m_\lambda\, \frac{d(\lambda,[0,\infty))}{8\, |\lambda|^{(1-\epsilon)/2 }\, |a|^{(1+\epsilon)/2}} \,.
\end{align*}
To summarize, we have shown that
$$
\sum_{\lambda\in Z}m_\lambda\, \frac{d(\lambda,[0,\ii))}{|\lambda|^{(1-\epsilon)/2}} \leq 8\, M_{N,q,\epsilon}\, |a|^{(1+\epsilon)/2} \,.
$$
The claimed inequality \eqref{eq:LT-1} follows by choosing $a<0$ such that $C_{N,q}^{\alpha_q} |a|^{(N-2q)\alpha_q/(2q)} \|V\|_{L^q}^{\alpha_q} =\epsilon_q$ (that is, equality holds in \eqref{eq:aconstraint}).
\end{proof}

\section{Application of the Limiting Absorption Principle: Schatten properties of the Scattering Matrix}\label{sec:SCATMAT}

An important object in scattering theory is the so-called scattering matrix. For every $\lambda>0$ (with the physical interpretation of an energy) this is a bounded operator $S(\lambda)$ on $L^2(\Sph^{N-1})$ (corresponding to the Fermi sphere in physics). Under rather weak assumptions on the potential $V$, the scattering matrix differs from the identity by a compact operator. In the next result we prove quantitative information in terms of trace ideals properties.

\begin{theorem}\label{thm:scatmat}
Let $N\ge 2$ and assume that $V\in L^q(\R^N,\R)$ with
$$
\begin{cases}
1 < q \le \frac32 & \text{if}\ N=2 \\
\frac N 2 \le q \le \frac{N+1}{2} & \text{if}\ N\ge 3 \,.
\end{cases}
$$
Then $S(\lambda)-1 \in \mathfrak S^{(N-1)q/(N-q)}(L^2(\Sph^{N-1}))$ for every $\lambda>0$ and the mapping $(0,\ii)\ni\lambda\mapsto S(\lambda)-1\in\gS^{(N-1)q/(N-q)}(L^2(\Sph^{N-1}))$ is continuous. Moreover, for all $\lambda>0$,
\begin{equation}
\label{eq:scatmatbound2}
\left\| S(\lambda) - 1 \right\|_{\mathfrak S^{\alpha_q}(L^2(\Sph^{N-1}))} \le C \lambda^{-1+N/2q} \|V\|_{L^q(\R^N)} \,, 
\end{equation}
where $\alpha_q:=\max(2,(N-1)q/(N-q))$ and $C$ is independent of $V$ and $\lambda$.
\end{theorem}

We will show in Remark \ref{rem:smopt} that for any $q$ as in the theorem the trace ideal $\mathfrak S^{(N-1)q/(N-q)}$ is best possible for the conclusions of the theorem to hold.

We do not know whether the need to consider $\alpha_q:=\max(2,(N-1)q/(N-q))$  instead of $(N-1)q/(N-q)$ in \eqref{eq:scatmatbound2} is technical or not, but in any case it only affects the parameter range $1<q<4/3$ in $N=2$. Moreover, in this range our proof below shows that \eqref{eq:scatmatbound2} holds with $(N-1)q/(N-q)=q/(2-q)$ provided $\lambda^{-1+1/q} \|V\|_{L^q(\R^2)} \leq M$ for any $M>0$ and with $C=C_M$ depending on $M$.

Previously, trace ideal properties of $S(\lambda)-1$ were mostly known under pointwise assumptions on $V$. (An exception are the bounds in \cite{sobolev-yafaev}, but there only the Hilbert--Schmidt norm of $S(\lambda)-1$ is considered.) Our contribution is to replace these pointwise assumptions by $L^q$ assumptions on $V$. For example, in \cite{kuroda} (see also \cite[Prop. 8.1.5]{Yafaev-10}) it is shown that, if $|V(x)| \le C (1+|x|^2)^{-\rho/2}$ for some $\rho>1$, then $S(\lambda)-1 \in \mathfrak S^r$ for any $r>(N-1)/(\rho-1)$. This follows also from our theorem provided $2N/(N+1)<\rho\leq 2$. Moreover, \cite[Thm. 8.2.1]{Yafaev-10} implies that $S(\lambda)-1 \not\in \mathfrak S^{(N-1)/(\rho-1)}$ for $V$ satisfying $V(x) \sim -c|x|^{-\rho}$ as $|x|\to\infty$ with $c\neq 0$. This fact shows that the trace ideal $\mathfrak S^{(N-1)q/(N-q)}$ in our theorem is optimal. Our argument in Remark \ref{rem:smopt} will be different.

\medskip

For the definition of the scattering matrix we refer to \cite{Yafaev-10}. The formula that will be important for us is that
\begin{equation}
\label{eq:scatmatrepr}
S(\lambda) = 1 - 2\pi i \Gamma_0(\lambda) \sqrt{|V|} \left( 1- \sqrt{V} (-\Delta +V-\lambda - i0)^{-1} \sqrt{|V|} \right) \sqrt{V} \Gamma_0(\lambda)^*
\end{equation}
(see (6.6.19) in \cite{Yafaev-10}). Here $\Gamma_0(\lambda)$ is the operator that maps functions $\psi$ on $\R^N$ to functions on $\Sph^{N-1}$ by restricting the Fourier transform to the sphere of radius $\sqrt\lambda$,
$$
\left(\Gamma_0(\lambda)\psi\right)(\omega) = 2^{-1/2} \lambda^{(N-2)/4} \widehat\psi(\sqrt{\lambda}\omega),\quad\forall\omega\in\Sph^{N-1} \,.
$$
Under the conditions of the theorem the products $\Gamma_0(\lambda) \sqrt{|V|}$ and $\sqrt{V} \Gamma_0(\lambda)^* = ( \Gamma_0(\lambda) \sqrt{V} )^*$ are bounded operators between the corresponding $L^2$ spaces by the Stein--Tomas restriction theorem (see also the proof below). Therefore the right side in the above formula for $S(\lambda)$ is well-defined.

The proof of \eqref{eq:scatmatbound2} given below was suggested to us by A. Pushnitski to whom we are grateful. It simplified and improved our original argument.

\begin{proof}[Proof of Theorem \ref{thm:scatmat}]
It follows from Theorem \ref{thm:schatten-WTSW} that $\mathcal E_{\Sph^{N-1}}^* V \mathcal E_{\Sph^{N-1}} \in \mathfrak S^{(N-1)q/(N-q)}(L^2(\Sph^{N-1}))$ with
$$
\left\| \mathcal E_{\Sph^{N-1}}^* V \mathcal E_{\Sph^{N-1}} \right\|_{\mathfrak S^{(N-1)q/(N-q)}(L^2(\Sph^{N-1}))} \le C_{N,q} \|V\|_{L^q(\R^N)}
$$
under the assumptions on $q$ in the theorem. Thus,
\begin{align*}
\left\| \mathcal E_{\Sph^{N-1}}^* \sqrt{|V|} \right\|_{\mathfrak S^{2(N-1)q/(N-q)}(L^2(\R^N),L^2(\Sph^{N-1}))}^2
& = \left\| \sqrt{|V|} \mathcal E_{\Sph^{N-1}} \right\|_{\mathfrak S^{2(N-1)q/(N-q)}(L^2(\Sph^{N-1}),L^2(\R^N))}^2 \\
& \leq C_{N,q} \|V\|_{L^q(\R^N)} \,.
\end{align*}
By scaling,
\begin{align*}
\left\| \Gamma_0(\lambda) \sqrt{|V|} \right\|_{\mathfrak S^{2(N-1)q/(N-q)}(L^2(\R^N),L^2(\Sph^{N-1}))}^2 & = \left\| \sqrt{|V|} \Gamma_0(\lambda)^* \right\|_{\mathfrak S^{2(N-1)q/(N-q)}(L^2(\Sph^{N-1}),L^2(\R^N))}^2 \\
& \leq 2^{-1} C_{N,q} \ \lambda^{-1+N/2q} \|V\|_{L^q(\R^N)} \,.
\end{align*}
Thus, it follows that
$$
\left\| \Gamma_0(\lambda) V \Gamma_0(\lambda)^* \right\|_{\mathfrak S^{(N-1)q/(N-q)}(L^2(\R^N))} \leq 2^{-1} C_{N,q} \ \lambda^{-1+N/2q} \|V\|_{L^q(\R^N)} \,.
$$
Moreover, by H\"older's inequality for trace ideals and Theorem \ref{thm:LAP} we have $\Gamma_0(\lambda) V (-\Delta+V - \lambda-i0)^{-1} V \Gamma_0(\lambda) \in \mathfrak S^{r}(L^2(\R^N))$ with $r=2(N-1)q/(3N-2q-1)$ (which satisfies $r<(N-1)q/(N-q)$), and for $|z|$ large enough we can bound the norm as follows,
\begin{align*}
& \left\| \Gamma_0(\lambda) V (-\Delta+V - \lambda-i0)^{-1} V \Gamma_0(\lambda) \right\|_{\mathfrak S^{r}(L^2(\R^N))} \\
& \le \left\| \Gamma_0(\lambda) \sqrt{|V|} \right\|_{\mathfrak S^{2(N-1)q/(N-q)}(L^2(\R^N),L^2(\Sph^{N-1}))}^2 \left\| \sqrt{V} (-\Delta+V - \lambda-i0)^{-1} \sqrt{|V|} \right\|_{\mathfrak S^{2q}(L^2(\R^N))} \\
& \le C_{N,q}' \lambda^{-2+N/q} \|V\|_{L^q(\R^N)}^2 \,.
\end{align*}
This proves that \eqref{eq:scatmatrepr} is well-defined and belongs to $1+ \gS^{(N-1)q/(N-q)}(L^2(\Sph^{N-1}))$. (It also proves the bound \eqref{eq:scatmatbound2} in the case $\lambda^{-1+N/2q} \|V\|_{L^q(\R^N)}\le C$, but we will give a different argument below that works for any $\lambda>0$.)

Because of the formula \eqref{eq:scatmatrepr} and the continuity statement in Theorem \ref{thm:LAP}, the continuity statement about the scattering matrix will follow if we prove that the mapping $(0,\ii)\ni\lambda\mapsto \Gamma_0(\lambda) W\in\gS^{2(N-1)q/(N-q)}(L^2(\R^N),L^2(\Sph^{N-1}))$ is continuous for $W\in L^{2q}(\R^N)$. This is well-known if $W$ is bounded and has compact support; see for example \cite[Lem. 8.1.2]{Yafaev-10}. The case of a general $W\in L^{2q}(\R^N)$ then follows similarly as in the proof of Lemma \ref{lem:extension-resolvent}. We decompose $W=W_1+W_2$ with $W_1$ bounded and compactly supported and $\|W_2\|_{L^{2q}(\R^N)}\leq \epsilon$ and use the a-priori bounds from Theorem \ref{thm:schatten-WTSW} to control the $W_2$ piece. This completes the proof of the continuity statement.

Finally, we turn to the proof of \eqref{eq:scatmatbound2}. In view of the identity \eqref{eq:decomp-resolvent} we can write the identity \eqref{eq:scatmatrepr} as
\begin{align*}
S(\lambda) & = 1 - 2\pi i \Gamma_0(\lambda) \sqrt{|V|} \left( 1+ \sqrt{V} (-\Delta -\lambda - i0)^{-1} \sqrt{|V|} \right)^{-1} \sqrt{V} \Gamma_0(\lambda)^* \\
& = 1 - 2\pi i \Gamma_0(\lambda) \sqrt{V} (\sgn V) \left( 1+ \sqrt{V} (-\Delta -\lambda - i0)^{-1} \sqrt{V} (\sgn V) \right)^{-1} \sqrt{V} \Gamma_0(\lambda)^* \,.
\end{align*}
This operator is of the form considered in \cite[Secs. 7.7 and 7.9]{Yafaev-92}. Therefore the abstract theorem \cite[Thm. 7.9.4]{Yafaev-92}, originally from \cite{sobolev-yafaev}, yields
$$
\left\| S(\lambda) - 1 \right\|_{\mathfrak S^{\alpha_q}(L^2(\Sph^{N-1}))} \le \kappa_{\alpha_q} \left\| \sqrt{V} (-\Delta-\lambda-i0)^{-1} \sqrt{V} \right\|_{\mathfrak S^{\alpha_q}(L^2(\R^{N}))}
$$
with $\kappa_p^p = 2^p \min_{0<\beta< 1} (\beta^{-p} + 2(1-\beta)^{-p})$. The proposition now follows from Theorem \ref{thm:uniform-sobolev-schatten} (if $N\ge 3$) and Proposition \ref{resbound2d} (if $N=2$).
\end{proof}

\begin{remark}\label{rem:smopt}
Let us prove optimality of the trace ideal $\mathfrak{S}^{(N-1)q/(N-q)}$ for any fixed $q$ as in the theorem. It follows from \eqref{eq:scatmatrepr} and the fact that $\sqrt{V} (-\Delta +V-\lambda - i0)^{-1} \sqrt{|V|}$ is compact by Theorem \ref{thm:LAP}, that $S(\lambda)-1\in\mathfrak S^r$ iff $\Gamma_0(\lambda) V \Gamma_0(\lambda)^* \in\mathfrak S^r$ for any fixed $r$. Therefore the optimality follows from Theorem \ref{thm:optimality-schatten-exponent} and the duality principle of Lemma \ref{lemma:duality-principle}. A similar argument (based on Knapp's example on the sphere) shows that not even the operator norm of $S(\lambda)-1$ can be bounded in terms of $\|V\|_{L^{q}(\R^N)}$ if $2q>N+1$.
\end{remark}

\medskip

\noindent\textbf{Acknowledgments.} The authors are grateful to A. Laptev, M. Lewin and A. Pushnitski for useful discussions. J. S. thanks the Mathematics Department of Caltech for the Research Stay during which this work has been done. Financial support from the U.S. National Science Foundation through grant PHY-1347399 (R. F.), from the ERC MNIQS-258023 and from the ANR ``NoNAP'' (ANR-10-BLAN 0101) of the French ministry of research (J. S.) are acknowledged.

\bibliographystyle{siam}

\end{document}